\newif\iffull
\fulltrue

\documentclass{article}

    \usepackage[nonatbib,final]{neurips_2025}

\usepackage[utf8]{inputenc} 
\usepackage[T1]{fontenc}    
\usepackage{hyperref}       
\usepackage{url}            
\usepackage{booktabs}       
\usepackage{amsfonts}       
\usepackage{nicefrac}       
\usepackage{microtype}      
\usepackage{xcolor}         
\usepackage{algorithm, algorithmic}
\usepackage[style=alphabetic,natbib=true,maxbibnames=9,maxalphanames=9,maxcitenames=9]{biblatex}
\usepackage{amsmath}
\usepackage{thmtools,thm-restate}
\usepackage{xspace}
\usepackage{enumitem}
\iffull \else
    \usepackage[skip=1pt]{caption}
\fi
\usepackage[colorinlistoftodos]{todonotes}
\definecolor{debianred}{rgb}{0.84, 0.04, 0.33}

\newcommand{\jnote}[1]{{\textcolor{blue}{Jakub: #1}}}
\newcommand{\smnote}[1]{{\textcolor{red}{Sepideh: #1}}}

\newtheorem{theorem}{Theorem}[section]
\newtheorem{lemma}[theorem]{Lemma}

\newtheorem{fact}[theorem]{Fact}

\newtheorem{definition}[theorem]{Definition}

\newenvironment{proof}{\par\noindent\textit{Proof.}}{$\Box$\par\bigskip\par}
\newenvironment{proofsketch}{\par\noindent\textit{Proof sketch.}}{$\Box$\par\bigskip\par}
\newenvironment{proofof}[1]{\par\noindent\textit{Proof of #1.}}{$\Box$\par\bigskip\par}
\newenvironment{proofsketchof}[1]{\par\noindent\textit{Proof sketch of #1.}}{$\Box$\par\bigskip\par}
\usepackage[capitalize, nameinlink]{cleveref}
\crefname{theorem}{Theorem}{Theorems}
\Crefname{fact}{Fact}{Facts}
\Crefname{lemma}{Lemma}{Lemmas}
\Crefname{claim}{Claim}{Claims}
\Crefname{observation}{Observation}{Observations}
\Crefname{invariant}{Invariant}{Invariants}

\newcommand{\bR}{\mathbb{R}}
\newcommand{\bE}{\mathbb{E}}
\newcommand{\bP}{\mathbb{P}}
\newcommand{\cI}{\mathcal{I}}
\newcommand{\cC}{\mathcal{C}}
\newcommand{\cU}{\mathcal{U}}
\newcommand{\OPT}{\operatorname{OPT}}
\newcommand{\Hypergeometric}{\operatorname{Hypergeometric}}
\newcommand{\OPTMatInt}{\OPT_{\mathrm{MatInt}}}
\newcommand{\eps}{\varepsilon}
\newcommand{\ceil}[1]{\lceil #1 \rceil}
\newcommand{\floor}[1]{\lfloor #1 \rfloor}
\newcommand{\rb}[1]{\left( #1 \right)}
\newcommand{\sqb}[1]{\left[ #1 \right]}
\newcommand{\fv}{\operatorname{fav}}

\newcommand{\infullversion}{in the full version (see the supplementary material)}

\title{Improved Algorithms for Fair Matroid\\Submodular Maximization \iffull \\ (full version)\fi}

\author{%
  Sepideh Mahabadi \\
  Microsoft Research\\
  \texttt{smahabadi@microsoft.com} \\
  \And
  Sherry Sarkar\thanks{Part of this work was done while an intern at Microsoft Research.}\\
  Carnegie Mellon University\\
  \texttt{sherrys@andrew.cmu.edu} \\
  \AND
  Jakub Tarnawski \\
  Microsoft Research\\
  \texttt{jakub.tarnawski@microsoft.com} 
}

\begin{document}

\maketitle

\begin{abstract}
Submodular maximization subject to matroid constraints is a central problem with many applications in machine learning. As algorithms are increasingly used in decision-making over datapoints with sensitive attributes such as gender or race, it is becoming crucial to enforce fairness to avoid bias and discrimination. Recent work has addressed the challenge of developing efficient approximation algorithms for fair matroid submodular maximization. However, the best algorithms known so far are only guaranteed to satisfy a relaxed version of the fairness constraints that loses a factor 2, i.e., the problem may ask for $\ell$ elements with a given attribute, but the algorithm is only guaranteed to find $\lfloor \ell/2 \rfloor$. In particular, there is no provable guarantee when $\ell=1$, which corresponds to a key special case of perfect
\iffull matching constraints. \fi

In this work, we achieve a new trade-off via an algorithm that gets arbitrarily close to full fairness. Namely, for any constant $\varepsilon>0$, we give a constant-factor approximation to fair monotone matroid submodular maximization that in expectation loses only a factor $(1-\varepsilon)$ in the lower-bound fairness constraint. Our empirical evaluation on a standard suite of real-world datasets --
\iffull including \fi
clustering, recommendation, and coverage tasks -- demonstrates the practical effectiveness of our methods.
\end{abstract}

The code for the paper is available at \url{https://github.com/dj3500/fair-matroid-submodular-neurips2025}.

\section{Introduction}

Machine learning is increasingly deployed in high-stakes decision-making, raising concerns about the propagation of bias and unfairness in automated systems. These challenges are especially acute in domains such as education, law enforcement, hiring, and credit~\cite{executive2016,blueprint22,reportEU22}. In response, a growing body of research has focused on developing algorithms that incorporate fairness constraints for core problems including clustering~\cite{Chierichetti0LV17}, data summarization~\cite{CelisKS0KV18}, classification~\cite{ZafarVGG17}, voting~\cite{CelisHV18}, and ranking~\cite{CelisSV18}. 

This paper studies fairness in the context of monotone submodular maximization subject to matroid constraints. Submodular functions, which capture the principle of diminishing returns, are fundamental to a range of machine learning applications such as recommender systems~\cite{El-AriniG11}, feature selection~\cite{DasK11}, active learning~\cite{GolovinK11}, and data summarization~\cite{LinB11}. Matroids provide a general framework for modeling independence constraints, encompassing cardinality, partition, graph connectivity, and linear independence constraints.

While numerous fairness definitions have been proposed, we adopt a widely used group fairness model, which partitions the universe into \emph{disjoint} groups
and
enforces \emph{lower and upper} bounds on the representation of each sensitive group in the selected set.
See \cref{sec:prelim-fmmsm} for a precise definition.
This model generalizes several fairness notions, such as proportional representation~\cite{monroe1995fully, Brill2017}, diversity constraints~\cite{cohoon2013, biddle2006adverse}, and statistical parity~\cite{dwork2012fairness}. It has been used for both submodular maximization \cite{CelisSV18, CelisHV18, HalabiMNTT20, el2023fairness, WFM21, tang2023beyond, yuan2023group, el2024} as well as a multitude of other optimization problems, such as clustering \cite{Chierichetti0LV17, KAM19, JNN20, HMV23}, voting \cite{CelisHV18}, data summarization \cite{CelisKS0KV18}, matching~\cite{Chierichetti0LV19} or ranking \cite{CelisHV18}.

In the absence of fairness constraints, monotone submodular maximization under a single matroid constraint is very well understood, as a tight $(1-1/e) \approx 0.63$-approximation is achievable~\cite{CalinescuCPV11,Feige98}.
The intersection of two matroid constraints
(which we refer to as ``matroid intersection'')
admits an almost $0.5$-approximation~\cite{Lee2010}.
The fair variant has been primarily explored under cardinality constraints~\cite{CelisHV18}, where a tight $(1-1/e)$-approximation is also known.
In the (single-pass) streaming setting, there is a $0.3178$-approximation~\citep{FeldmanLNSZ22}
for the non-fair matroid version;
furthermore, since the intersection of cardinality constraint and fairness can be reduced to a single matroid constraint~\cite{HalabiMNTT20},
the same approximation factor can be obtained for it.

However, the intersection of a matroid constraint and a fairness constraint
seems significantly more challenging,
and is still poorly understood despite
two recent works devoted to studying this problem
in the streaming~\cite{el2023fairness} and the classic offline~\cite{el2024} settings;
our focus is on the latter.
Following \cite{el2023fairness}, we refer to the problem  as Fair Matroid Monotone Submodular Maximization (\textbf{FMMSM}).
To appreciate its difficulty,
consider a key special case,
Monotone Submodular Perfect Matching (\textbf{MSPM}),
i.e., maximizing a monotone submodular function
over the collection of all \emph{perfect} matchings in a \emph{bipartite} graph $(V_G,E_G)$.\footnote{To see why MSPM is a special case of FMMSM, set $E_G$ as the universe, consider a partition matroid that encodes that every vertex on the left shall have degree at most 1 in the solution, and set fairness constraints so that every vertex on the right shall have degree at least 1 and at most 1.}
This collection of feasible sets is not downward-closed,
which invalidates known algorithmic approaches.\footnote{Of course, a proper subset of a perfect matching is not a perfect matching. But more importantly, the collection of all \emph{subsets of perfect matchings} (which is downward-closed) does not belong to any of the families that are known to make approximate submodular maximization tractable. In particular, it is not a matroid, an intersection of
\iffull a small number of \fi
matroids, or a so-called $p$-extendible set system or a $p$-system~\cite{CalinescuCPV11} for $p=O(1)$.}
The best known approximation factor
for MSPM is a trivial $O(|V_G|)$-approximation;
one can also apply the framework of~\cite{GoemansHIM09} to obtain an
\iffull
$\widetilde{O}(\sqrt{|E_G|})$-approximation\footnote{
    The work~\cite{GoemansHIM09} shows that we can in polynomial time obtain numbers $c_e$ for $e \in E_G$
    such that for any $S \subseteq E_G$, the function $\hat{f}(S) := \sqrt{\sum_{e \in S} c_e}$ is an $\widetilde{O}(\sqrt{|E_G|})$-approximation to $f(S)$.
    Maximizing $\hat{f}(S)$ amounts to maximizing $\sum_{e \in S} c_e$,
    which is the maximum-weight bipartite perfect matching problem,
    solvable in polynomial time.
},
\fi
which is superior for sparse graphs.
In fact, this could possibly even be tight,
as it
almost matches a surprising
negative result of~\cite{el2024}
who showed a family of sparse graphs
where the standard \emph{multilinear relaxation}
(commonly used in relax-and-round approaches for submodular optimization)
has an integrality gap of $\Omega(\sqrt{|E_G|})$.
The existence of a constant-factor approximation
to MSPM was posed by~\cite{el2024} as an exciting open problem.

The algorithms given in~\cite{el2023fairness,el2024} for FMMSM
circumvent the difficulty posed by the lower bound constraints by relaxing them.
They obtain the following two results:
\begin{theorem}[Two-pass algorithm of \cite{el2023fairness}] \label{thm:streaming}
    There is a polynomial-time algorithm for FMMSM that violates lower bound constraints by a factor 2 and obtains $\alpha/2$-approximation, where $\alpha$ is the approximation ratio of an algorithm for maximizing a monotone submodular function under a matroid intersection constraint.
\end{theorem}
We can have $\alpha$ be almost $1/2$~\cite{Lee2010} and thus get an almost $1/4$-approximation.
(\cite{el2023fairness} work in the streaming setting and instead use the streaming algorithm for matroid intersection of~\cite{GargJS21}; this results in a $1/11.66$-approximation in two passes.)
Here, violating lower bound constraints by a factor 2 means that,
if a color has a lower bound of $\ell$,
the solution is guaranteed to have at least $\floor{\ell/2}$ elements of that color.
Note that in MSPM we have $\ell=1$ and thus $\floor{\ell/2}=0$.
\begin{theorem}[\cite{el2024}] \label{thm:inexp}
    There is a polynomial-time algorithm for FMMSM that satisfies \emph{lower and upper} bound constraints in expectation rather than exactly,
    and obtains a $(1-1/e)$-approximation in expectation.
\end{theorem}
\cref{thm:inexp} also guarantees certain two-sided tail bounds on the violation of each fairness constraint which apply if $\ell$ is large enough.
It is the only algorithm considered in this paper that violates the \emph{upper} bounds.
The algorithm proceeds by solving and rounding the multilinear relaxation.

If we consider a relaxed version of MSPM
where instead of a \emph{perfect} matching
we want a \emph{large} matching that also has high submodular function value,
then a simple greedy algorithm
will yield a $1/3$-approximation (\cref{thm:mi-greedy})
and construct a maximal matching,
thus getting $1/2$ of the maximum possible size.
The results in \cref{thm:streaming,thm:inexp} give no improvement upon
\iffull
this.\footnote{The result of \cref{thm:inexp} (\cite{el2024}) only satisfies the upper bounds in expectation, so to obtain a feasible solution for MSPM (i.e., a matching), one needs to delete some violating edges from the solution, which will damage the objective value and the solution cardinality.
Therefore, as stated, the result gives no theoretical guarantees for MSPM.
By opening up the algorithm and its proof, one can use the bounds for randomized swap rounding~\cite[Theorem II.3]{Chekuri2010} to get some meaningful bound: $\bP[|S \cap \delta(v)| \ge 1] \ge 1 - e^{-1/2} \approx 0.4$. This means that, for a bipartite graph with $n+n=2n$ vertices, $0.4n$ right-side vertices will be guaranteed to have an incident edge in $S$; when we select one edge per right-side vertex, we can expect a matching of size $0.4n$, and one can perhaps hope for a similar bound on submodular value, i.e., $0.4(1-1/e)\OPT \approx 0.25 \OPT$. But a simple greedy algorithm gets $0.5n$ and $\textsc{OPT}/3$, respectively.}
\fi
While one can try to generalize this simple approach to FMMSM,
it faces another issue that is salient in the context of fairness motivations:
while at least half of the total lower bound mass
will be satisfied,
there could be ``unlucky'' colors
(marginalized groups)
that never get represented in the solution;
this is precisely the reason why we seek fair algorithms in the first place.

\subsection{Our contributions}

In this work we
provide an algorithm that
satisfies the fairness constraints
within a factor better than 2,
while also giving guarantees for every individual group (rather than only in aggregate like the simple greedy strategy discussed above).
To achieve the
\iffull former objective, \fi
we trade off part of the objective value;
to achieve the latter, we employ randomization.

\begin{theorem}[informal version of \cref{thm:main-randomized}]
    For every $\eps \in (0,1)$ there is a polynomial-time algorithm for FMMSM whose output
    \iffull
    \begin{itemize}
    \else
    \begin{itemize}[topsep=0pt,itemsep=0pt]
    \fi
        \item satisfies the matroid constraint,
        \item satisfies fairness upper bound constraints,
        \item 
        for a group with fairness lower bound $\ell$,
        has in expectation at least $(1-\eps)\ell$ elements from that group,
        \item has expected size at least $(1-\eps)$ times the maximum size of any feasible solution,
        \item satisfies Chernoff-style high-probability bounds on size, as well as total fairness violation,
        \item has expected submodular function value at least $0.499 \cdot \eps \cdot \OPT$.
    \end{itemize}
\end{theorem}
Our bound on the submodular function value is actually shown with respect to a more powerful optimum, namely, an optimal set that satisfies the matroid and upper-bound constraints, but not necessarily the lower-bound constraints.
If one wants to compare to this optimum,
then the $O(\eps)$ factor loss in value is unavoidable.
To see this, consider MSPM in a graph $P_3 \times N$ consisting of a disjoint union of $N$ paths of length 3, with a linear objective function assigning values $0,1,0$ to each path's edges. A perfect matching of size $2N$ has $0$ value, and a maximal matching of size $N$ has value $N$; one can interpolate between these smoothly.

We note that by instantiating $\eps=1/2$ we obtain an almost $1/4$-approximation
while violating lower bounds by a factor $2$,
which is similar to the bounds of \cref{thm:streaming} (\cite{el2023fairness}).

As a second contribution,
we also employ our techniques to obtain a deterministic algorithm.
There are several variants that we could formulate;
we choose to show a general setting of matroid intersection,
where the trade-off is between size and objective value.
The relation to fairness is that an algorithm that finds a solution of maximum size that is an $\alpha$-approximation to the objective value would imply an $\alpha$-approximation algorithm for FMMSM
(see \cite{el2023fairness}, Proposition C.6).
\iffull
\begin{theorem}[informal version of \cref{thm:main-deterministic}]
\else
\begin{theorem}
\fi
\label{thm:deterministic-informal}
    For every $\eps \in (0,1)$
    there is a deterministic polynomial-time algorithm for the problem of maximizing a monotone submodular function subject to two matroid constraints whose output has size at least $(1-\eps)$ times the maximum size of any feasible solution minus one,
    and obtains a $(0.499 \cdot \eps)$-approximation to the submodular function value.
\end{theorem}

\iffull
\paragraph{Experimental results.}
\fi
We show the effectiveness of our algorithm
empirically
against prior work and natural baselines
on a suite of standard benchmarks.
We measure the submodular objective value and total fairness violation.
Our algorithms produce solutions whose value is  competitive with the highest-value baseline, which completely ignores the lower bound constraints
and accordingly has the highest fairness violations.
In two out of three scenarios,
our algorithms dominate prior work~\cite{el2023fairness}.
Finally, a key strength of our approach is the flexibility given by $\eps$, allowing users to tune the balance between utility and fairness.

\iffull
\paragraph{Our techniques.}
\fi
Let us begin with the simple setting of perfect matchings (MSPM).
Consider the symmetric difference of a high-value matching $Y$ and a perfect matching
\iffull $P$
(where $Y$ might be small and $P$ might have no value). \fi
This decomposes into a collection of vertex-disjoint alternating cycles and augmenting paths.

One possible algorithm is to ignore the cycles, and choose some of the augmenting paths to apply to~$Y$,
so that its size grows to at least $(1-\eps)|P|$.
We can do this by computing the marginal contribution of the elements that $Y$ would lose in each path, and taking the least damaging paths;
by submodularity, this loses at most a $(1-\eps)$ fraction of value in $Y$.

While this does ensure a large matching,
some $\eps$ fraction of vertices can still be ``unlucky'' and end up unmatched.
Deterministically this would be hard to avoid
(short of solving MSPM/FMMSM completely, with no fairness violation);
our next idea is to choose the paths randomly in the above solution.
This will work for MSPM, as long as we take care to select a $(1-\eps)$ fraction of the $|P|-|Y|$ many augmenting paths,
even if we already have $|Y| \ge (1-\eps)|P|$.
Then every vertex that was not matched in $Y$
has a $(1-\eps)$ probability of being matched in the new solution.

However, there are two main challenges
when trying to generalize the above approach
to matroid and fairness constraints.
Firstly,
having fairness bounds with $\ell_c < u_c$ means that
$Y$ can have fewer elements than $P$ in some colors
but more elements in other colors,
and can even have $|Y|=|P|$ while still violating many fairness lower bounds.
This means that we need to find and apply not only
augmenting paths, but also alternating paths that exchange an element of an oversaturated color for one of an undersaturated color,
without increasing the solution size.
We show that as long as the total fairness violation is large, there are many such disjoint paths,
which implies that applying a random fraction of them still retains enough value.

The second, larger obstacle
arises due to dealing with general matroids.
We are able to use tools from matroid theory
to show the existence of many disjoint alternating or augmenting paths
in an appropriate matroid intersection exchange graph
whose vertices correspond to elements of $Y$ and $P$
(which were edges in the case of MSPM).
We need to carefully refine the paths
via an asymmetric shortcutting process
to ensure that applying them leaves the solution independent in the matroid
while also not disrupting the counts of elements in the colors not being exchanged.
Moreover, in general, multiple augmenting paths in matroids cannot be applied simultaneously.
We deal with this using an iterative framework where we apply a single path, rebuild the exchange graph, and find a new large collection of disjoint paths;
we then bound the loss in value after each step.

\iffull
\paragraph{Paper organization.}
We discuss more related work in \cref{sec:addtl-related-work}.
\fi
In \cref{sec:prelim} we introduce all necessary notation, definitions, and useful facts.
In \cref{sec:our-algo} we describe our algorithms and prove their properties.
\cref{sec:experiments} is devoted to the experimental evaluation.
We conclude and discuss the limitations and broader impact of our work in \cref{sec:conclusion}.

\iffull
\subsection{Additional related work}
\label{sec:addtl-related-work}
The {\em non-monotone} Fair Matroid Submodular Maximization problem was explored by
\cite{yuan2023group}
under the cardinality constraint. They achieved a $0.2005$-approximation for the special case where for all colors $c$, we have $\ell_c/|V_c|=a$ and $u_c/|V_c|=b$ for some constants $a,b\in [0,1]$. Later, \cite{el2024} recovered and further generalized their results. In particular for general matroids, they achieved a $(1-\beta)/(8+\eps)$ approximation algorithm that guarantees the number of elements from each group $c$ is between $\floor{\beta\ell_c}$ and $u_c$ for a trade-off parameter $\beta\in [0,1/2]$.

In this work, we consider the setting where the color groups are disjoint. The more general case, where groups may overlap, was previously studied by
\cite{CelisHV18}
for the special case of FMMSM with a cardinality constraint. They show that when elements can belong to three or more groups, simply checking the feasibility becomes NP-hard. However, by allowing for violations of the fairness constraints and in particular guaranteeing the fairness constraint in expectation, they gave a $(1-1/e-o(1))$-approximation algorithm for the problem.

An alternative notion of fairness in submodular maximization has been explored in
\cite{tsang2019,tang2023beyond,Wang2024},
where the focus is on ensuring that each group -- potentially not limited to subsets of the ground set $V$ -- receives at least a specified amount of value from the selected solution. In these formulations, the value is  modeled using a monotone submodular function. This line of work can be cast as a multi-objective submodular maximization problem
\cite{Krause2008,Chekuri2010,Udwani2018}.
\fi

\section{Preliminaries}
\label{sec:prelim}

We denote the symmetric difference $(X \setminus Y) \cup (Y \setminus X)$ of two sets $X$ and $Y$ by $X \triangle Y$.

\iffull
\paragraph{Submodular functions.}
\fi
We consider functions $f: 2^V \to \bR_+$ defined on a ground set $V$.
We say that $f$ is \emph{submodular} if 
   $f(Y \cup \{e\}) - f(Y) \ge f(X \cup \{e\}) - f(X)$
   for any two sets $Y \subseteq X \subseteq V$ and any element $e\in V \setminus X$. Moreover, $f$ is  \emph{monotone} if $f(Y) \leq f(X)$ for any two sets $Y \subseteq X \subseteq V$. We assume that $f$ is given as an oracle that computes $f(S)$ for given $S \subseteq V$;
    we consider the running time of this oracle to be $O(1)$.

The following fact is folklore.
\iffull
    We provide a proof for completeness.
\fi
\begin{fact} \label{fact:submod_k}
    Let $f$ be a non-negative submodular function and $X_1, X_2, ..., X_k \subseteq X$ be disjoint subsets of $X$. Then
    \[ \sum_{i=1}^k f(X \setminus X_i) \ge (k-1) f(X) . \]
\end{fact}
\iffull
\begin{proof}
    We use induction on $k$. The base case $k=1$, i.e., that $f(X \setminus X_1) \ge 0$, follows because $f \ge 0$.
    For $k > 1$, we apply the inductive hypothesis to the set family $X_1 \cup X_2, X_3, X_4, ..., X_k$. We get
    \[ f(X \setminus (X_1 \cup X_2)) + \sum_{i=3}^k f(X \setminus X_i) \ge (k-2) f(X) . \]
    By submodularity,
    \[ f(X \setminus X_1) + f(X \setminus X_2) - f(X \setminus (X_1 \cup X_2)) \ge f(X) .\]
    Adding these two inequalities gives the statement.
\end{proof}
\fi

\iffull
\paragraph{Matroids.}
\fi
    A \emph{matroid} is a set family $\cI \subseteq 2^V$ 
    with the properties:
    \iffull
    \begin{itemize}
    \else
    \begin{itemize}[topsep=0pt,itemsep=0pt]
    \fi
        \item 
            \textit{Downward-closedness}: 
            if $X \subseteq Y$ and $Y \in \cI$, then $X \in \cI$; 
        \item    
            \textit{Augmentation}: 
            if $X, Y \in \cI$ and $|X| < |Y|$, then there exists $e \in Y$ with $X + e \in \cI$.
    \end{itemize}
    We abbreviate $X \cup \{e\}$ as $X+e$ and $X \setminus \{e\}$ as $X-e$.
    We assume that the matroid is given as an oracle
    that, for a given $S \subseteq V$,
    answers whether $S \in \cI$;
    we consider the running time of this oracle to be $O(1)$.
    We say that a set $S \subseteq V$ is \emph{independent} if $S \in \cI$.

\iffull
\paragraph{Matroid exchange graph.}
\fi
Let $\cI$ be a matroid on universe $V$
and $Y$, $Z$ be two independent sets.

\begin{definition}
     We define the exchange graph for $Y$ and $Z$ with respect to $\cI$ as the bipartite graph
    \[ (Y \setminus Z, Z \setminus Y, \{(y,z) : Y-y+z \in \cI\}) . \]
 \end{definition}

\begin{lemma}[\cite{Schrijver}, Corollary 39.12a] \label{lem:perfect-matching}
    If $|Y|=|Z|$,
    then
    the exchange graph for $Y$ and $Z$ with respect to $\cI$
    contains a perfect matching.
\end{lemma}

\iffull
\begin{lemma}[\cite{Schrijver}, Corollary 39.13] \label{lem:unique-perfect-matching}
    Let $Y$ be an independent set and let $Z \subseteq V$ be such that $|Z| = |Y|$. If the exchange graph for $Y$ and $Z$ with respect to $\cI$
    contains a \textit{unique} perfect matching between $Y \setminus Z$ and  $Z \setminus Y$, then $Z$ is also an independent set. 
\end{lemma}
\fi





\subsection{Fair Matroid Monotone Submodular Maximization (FMMSM)} \label{sec:prelim-fmmsm}
The universe $V$ is partitioned into $C$ sets: $V = V_1 \cup V_2 \cup ... \cup V_C$, where $V_c$ denotes elements of color $c$.
Every element has exactly one color.
The set of colors is denoted by $[C] = \{1,2,...,C\}$.
For every color $c \in [C]$ we have \emph{fairness bounds}:
lower bound $\ell_c$ and upper bound $u_c$.

The set of upper bounds gives rise to a \emph{partition matroid} that we will denote by $\cU$.
That is, \[ \cU = \{ S \subseteq V \mid |S \cap V_c| \leq u_c \; \ \forall c \in [C]\}. \]
It is well-known that such a collection of sets forms a matroid.
We will call a set $S \in \cU$ \emph{upper-fair}.

If a set satisfies both the lower and the upper bounds, we say that it is \emph{fair}.
That is, we define the family of fair sets $\cC$ as follows:
\[ \cC = \{S \subseteq V \mid  \ell_c \leq |S \cap V_c| \leq  u_c \; \ \forall c \in [C]\} \, .\]

The FMMSM problem asks to find a set $S \in \cI \cap \cC$
(i.e., fair and independent $S$)
that maximizes $f(S)$.
We use $\OPT$ for the optimal value, i.e., $\OPT = \max_{S \in \cI \cap \cC} f(S)$.
We assume that there exists a fair and independent set, i.e., $\cI \cap \cC \neq \emptyset$.
We say that an algorithm is an $\alpha$-approximation
if it outputs a set
$S$
with $f(S) \ge \alpha \cdot \OPT$.

For any set $S \subseteq V$ we define its
\emph{fairness violation} $\fv(S) := \sum_{c} \max\{|S \cap V_c| - u_c, \ell_c - |S \cap V_c|, 0\}$.
Note that if $S$ is upper-fair,
then $\fv(S) = \sum_c \max\{\ell_c - |S \cap V_c|, 0\}$.

\begin{lemma}[\cite{el2023fairness}, Appendix C] \label{lem:polytime-linear-fmmsm}
    There is an exact polynomial-time algorithm for FMMSM
    for the case when $f$ is a linear function.
\end{lemma}

\iffull
\paragraph{Matroid intersection.}
\fi
Given two matroids
and a monotone submodular function $f$ defined on $V$,
we can define the problem of maximizing a submodular function
subject to a matroid intersection constraint similarly to FMMSM.
\iffull

\fi
In particular, if we ignore the lower bounds completely, FMMSM turns into the above matroid intersection problem for matroids $\cI$ and $\cU$.
\begin{theorem}[\cite{CalinescuCPV11}]
\label{thm:mi-greedy}
    The greedy algorithm gives a $1/3$-approximation to this problem.
\end{theorem}
\begin{theorem}[\cite{Lee2010}] \label{thm:local_search}
    For any $\delta > 0$
    there is a polynomial-time algorithm
    that gives a $(0.5-\delta)$-approximation
    to this problem.
\end{theorem}

\section{Our algorithm}
\label{sec:our-algo}

In this section we describe our algorithms:
randomized (\cref{thm:main-randomized}) and deterministic
\iffull
(\cref{sec:deterministic}, \cref{thm:main-deterministic}).
See also the pseudocode provided in \Cref{alg:main-randomized}.
We first need to introduce some notions.

\fi
The proof of \cref{thm:main-randomized}
will begin by constructing a maximum-cardinality
independent and fair set $P$,
which will stay unchanged throughout the execution.
We also construct an independent and upper-fair set $Y$ of high $f$-value.
We will use $P$ as a source of fairness and iteratively trade off $Y$'s value for $P$'s elements in colors that are undersaturated by $Y$.

\begin{definition}
    Given $Y$ and $P$ as above,
    we say that a color $c \in [C]$ is \emph{undersaturated} if $|Y \cap V_c| < |P \cap V_c|$, and \emph{oversaturated} if $|Y \cap V_c| > |P \cap V_c|$.
\end{definition}

The technical crux of the proof of \cref{thm:main-randomized} is \cref{lem:exist-k-paths}, in which we show the existence of many disjoint structures, each of which can be used to advance our fairness objective.
\iffull
We will call them augmenting or alternating,
as they indeed correspond to such paths
in the appropriately defined matroid intersection exchange graph that we consider in the proof of \cref{lem:exist-k-paths}.
\fi

\begin{definition} \label{def:paths}
    Let $Y$ be an independent and upper-fair set,
    and let $X \subseteq V$.
    Define the result $Y'$ of applying $X$ to $Y$ as the symmetric difference $Y' = Y \triangle X$.
    We say that
    $X$ is \emph{alternating} (with respect to $Y$) if
    $Y'$ is independent ($Y' \in \cI$)
    and
    there is exactly one undersaturated color $c' \in [C]$ and one oversaturated color $c'' \in [C]$ such that for all $c \in [C]$,
        \[
        |Y' \cap V_c| = |Y \cap V_c| + \begin{cases}
            1 & \text{ for $c = c'$,} \\
            -1 & \text{ for $c = c''$,} \\
            0 & \text{ for $c \ne c',c''$.}
        \end{cases}
        \]
    
    We say that $X$ is \emph{augmenting} if all the above conditions are satisfied, except that there is no color $c''$.
    \iffull

    \fi
    In both cases, we say that $X$ \emph{increases $c'$}.
\end{definition}
Note that we have $|Y'|=|Y|$ if $X$ is alternating and $|Y'|=|Y|+1$ if $X$ is augmenting. Also, $Y'$ is upper-fair, since the only color where it has more elements than $Y$ is $c'$, and we have $|Y' \cap V_{c'}| = |Y \cap V_{c'}| + 1 < |P \cap V_{c'}| + 1$ (and $P$ is fair).

\begin{lemma}
    \label{lem:exist-k-paths}
    Let $Y$ and $P$ be two independent and upper-fair sets
    with $|Y| \le |P|$.
    Denote
    \[ k = \sum_{c \in [C]} \max(0, |P \cap V_c| - |Y \cap V_c|) \,. \]
    Then we may find in polynomial time a collection $X_1, ..., X_k$ of disjoint subsets of $Y \cup P$,
    of which at least $|P|-|Y|$ many are augmenting
    and the
    rest
    are alternating.
    Moreover, for every undersaturated color $c$, exactly $|P \cap V_c|-|Y \cap V_c|$ many of the paths increase $c$.
\end{lemma}
\iffull
\begin{proof}
    To simplify notation, we assume without loss of generality that $Y \cap P = \emptyset$. Otherwise we could work with sets $Y \setminus P$ and $P \setminus Y$.

    Consider the matroid intersection exchange graph for $Y$ and $P$ with respect to matroids $\cI$ and $\cU$.
    This is defined as the \emph{directed} bipartite graph obtained by taking the union of the exchange graph for $Y$ and $P$ with respect to $\cI$, whose edges we direct right-to-left (from $P$ to $Y$), and of the exchange graph for $Y$ and $P$ with respect to $\cU$, whose edges we direct left-to-right (from $Y$ to $P$). That is, we have edges
    \begin{align*}
        \{ y \to p : Y + p - y \in \cU \}
        \quad \text{and} \quad
        \{ y \leftarrow p : Y + p - y \in \cI \} \,.
    \end{align*}
    Inside this graph we will carefully construct a subgraph consisting of two matchings $M_\to$ (directed left-to-right) and $M_{\leftarrow}$ (directed right-to-left). The augmenting and alternating paths will be found in that subgraph.

    To construct $M_\leftarrow$, we first define $T_P := \{ p \in P : Y + p \in \cI \}$.
    We have $|T_P| \ge |P|-|Y|$ (by repeated application of the matroid augmentation property).
    We will call the elements in $T_P$ \emph{$P$-sinks}.
    Let $T_P'$ be an arbitrary subset of $T_P$ of size exactly $|P|-|Y|.$
    We then invoke \cref{lem:perfect-matching} on the exchange graph for $Y$ and $P \setminus T_P'$ with respect to $\cI$ (which is a subgraph of our matroid intersection exchange graph).
    It implies the existence of
    a perfect matching between $Y$ and $P \setminus T_P'$;
    since one exists, we can find one in polynomial time.
    We obtain the matching $M_\leftarrow$
    by removing the edges of that matching that are incident to $T_P \setminus T_P'$.
    Then, $M_\leftarrow$ matches every vertex of $P \setminus T_P$.

    We construct the matching $M_\to$ manually by matching up as many elements of the same color between $Y$ and $P$ as possible. That is, for every color $c \in [C]$
    we add $\min(|Y \cap V_c|, |P \cap V_c|)$ edges from $Y \cap V_c$ to $P \cap V_c$ to the matching $M_\to$.

    We define the set $S$ of \emph{sources} as all vertices in $P$ that did not get matched in $M_\to$. Note that they are only in undersaturated colors, and their number is exactly $k$.
    (In principle it is possible to have a source that is also a $P$-sink;
    this can happen if $Y$ is not maximal in $\cI \cap \cU$.)

    We also define the set $T_Y$ of \emph{$Y$-sinks} as all vertices in $Y$ that did not get matched in $M_\to$.
    Note that they are only in oversaturated colors, and their number is exactly $k-(|P|-|Y|)$, as we have
    \begin{align*}
        |P| - |Y| &= \sum_{c \in [C]} |P \cap V_c| - |Y \cap V_c| \\
        &= \sum_{\text{$c$: undersaturated}} \left( |P \cap V_c| - |Y \cap V_c| \right) - \sum_{\text{$c$: oversaturated}} \left( |Y \cap V_c| - |P \cap V_c| \right) \\
        &= k - |T_Y| .
    \end{align*}
    To recap, we have $k$ sources $S$ (all in $P$),
    $k-(|P|-|Y|)$ $Y$-sinks $T_Y$,
    and
    at least
    $|P|-|Y|$ $P$-sinks $T_P$.
    Moreover, for every undersaturated color $c$, exactly $|P \cap V_c|-|Y \cap V_c|$ many of the 
    sources are of color $c$.

    Now we show how to construct $k$ vertex-disjoint simple paths in $M_\to \cup M_\leftarrow$ that start at sources ($S$) and end at sinks ($T_Y \cup T_P$). For every path, we proceed
    as follows:
    \begin{itemize}
        \item start at at unused source (in $P$),
        \item whenever at a vertex of $P$, stop if that vertex is a sink (in $T_P$); otherwise it has an incident outgoing edge of $M_\leftarrow$; follow this edge,
        \item whenever at a vertex of $Y$, stop if that vertex is a sink (in $T_Y$); otherwise it has an incident outgoing edge of $M_\to$; follow this edge.
    \end{itemize}
    Since every path must terminate at a different sink,
    at least $k - (k-(|P|-|Y|)) = |P|-|Y|$ of the $P$-sinks will be used.
    Furthermore, a path cannot revisit a vertex, since the indegree of every vertex is at most 1 and sources have no incoming edges.
    This implies that all paths are simple and vertex-disjoint.

    The $k$ paths constructed above might not yet be augmenting/alternating paths in the matroid intersection exchange graph, as they may contain chords;
    in general, only chordless paths guarantee that applying them preserves independence.
    (Matroid intersection algorithms usually apply shortest paths, which are chordless.)
    We need to shortcut them; however, doing so naively could destroy the property that all left-to-right edges in the paths are between elements of the same color, which we require to satisfy the condition in \cref{def:paths}.

    We carry out the shortcutting as follows. Let $X' = (p_1, y_1, p_2, y_2, ...)$ be one of the $k$ paths.
    As long as there exists a chord of the form $(p_i, y_j)$ with $j>i$ (i.e., the directed edge $y_j \leftarrow p_i$ exists in the matroid intersection exchange graph; equivalently, $Y + p_i - y_j \in \cI$), replace the corresponding subpath with this chord (i.e., remove the vertices $y_i, p_{i+1}, ..., p_j$ from the sequence $X'$).
    Note that we do not use chords of the form $(y_i, p_j)$ with $j>i$.
    Doing this to each of the $k$ paths obtains our final collection $X_1, ..., X_k$.
    
    We now verify that it satisfies the statement of the lemma.
    As the paths before shortcutting were vertex-disjoint,
    they remain so afterwards.
    We claim that the paths ending at $P$-sinks
    (recall that there are at least $|P|-|Y|$ many)
    yield augmenting sets,
    and the paths ending at $Y$-sinks
    yield alternating sets.
    Consider a path $X_i = (p_1, y_1, p_2, y_2, ...)$.
    Note that $Y' = Y \triangle X_i = Y \cup \{p_1,p_2,...\} \setminus \{y_1,y_2,...\}$.
    The color-count condition of \cref{def:paths} follows easily from the property that every left-to-right edge $y_i \to p_{i+1}$ in $X_i$ belongs to $M_\to$, so $y_i$ and $p_{i+1}$ are of the same color.
    Thus we can take $c'$ to be the color of $p_1$.
    If the last element of $X_i$ is in $Y$ (a $Y$-sink),
    we take $c''$ to be its color.

    It remains to show that $Y' = Y \triangle X_i \in \cI$. This argument closely follows that of \cite{Schrijver}, Theorem 41.2. Let us first consider the case where $X_i$ ends at a $Y$-sink: $X_i = (p_1, y_1, p_2, y_2, ..., p_t, y_t)$. We want to apply \Cref{lem:unique-perfect-matching} on the exchange graph for $Y$ and $Y'$ with respect to $\cI$. This is a bipartite graph on $\{y_1, ..., y_t\}$ on the left side and $\{p_1, ..., p_t\}$ on the right side,
    and it is equal to the corresponding induced subgraph of edges going right-to-left in the matroid intersection exchange graph; we need to show that it contains a unique perfect matching.
    We proceed iteratively: $p_1$ cannot be connected to any $y_j$ with $j>1$, for otherwise we would have a shortcut. So any matching must have $p_1$ matched to $y_1$. Removing these two vertices, we consider the out-neighborhood of $p_2$. Again, $p_2$ has no shortcuts to later $y_j$, so its only out-neighbor (after the removal of $y_1$) is $y_2$. So, $p_2$ must be matched to $y_2$. We may continue inductively to construct the unique matching between $\{y_1, ..., y_t\}$ and $\{p_1, ..., p_t\}$.
    
    The case where $X_i = (p_1, y_1, p_2, y_2, ..., p_t, y_t, p_{t+1})$ ends at a $P$-sink
    is similar, with one more step. The first case shows that $Z = Y \cup \{p_1, ..., p_t \} \setminus \{y_1, ..., y_t\}$ is independent. We need only show that $Z + p_{t+1}$ is independent. Note that $p_{t+1} \in T_P$, meaning that $Y + p_{t+1} \in \mathcal{I}$.
    By the matroid augmentation property, $Y + p_{t+1}$ must have an element which can be added to $Z$ while preserving independence.
    The possible candidates are $(Y + p_{t+1} ) \setminus Z = \{y_1,...,y_t,p_{t+1}\}$.
    However, no $y_j$ can be added:
    since $p_1, ..., p_t \not \in T_P$, we know that $Y \cup \{p_1, ..., p_t\}$ has rank $|Y|$,
    and $Z + y_j \subseteq Y \cup \{p_1, ..., p_t\}$
    would have rank $|Z|+1 = |Y|+1$ if $Z+y_j$ were independent.
    Therefore the only possible candidate is $p_{t+1}$, and so we have that $Y' = Z +p_{t+1}$ is independent. 
\end{proof}
\fi
\iffull
Now we are ready to state and prove our main result.
\fi
\begin{theorem} \label{thm:main-randomized}
    There is a randomized polynomial-time algorithm for FMMSM
    parametrized by $\eps \in (0,1)$
    that outputs a set $S \in \cI \cap \cU$ (i.e., independent and upper-fair) such that
    \iffull
    \begin{itemize}
    \else
    \begin{itemize}[topsep=0pt,itemsep=0pt]
    \fi
        \item $\bE[|S|] \ge (1-\eps) N$
        with a high-probability tail bound:
        \\
        for $\delta>0$,
        $\bP[|S| < (1-\delta) (1-\eps) N] \le \exp(-\Omega_\delta(N))$
        \item $\bE[f(S)] \ge 0.499 \cdot \eps \cdot \OPTMatInt$
        \item for every $c \in [C]$ we have
        $\bE[|S \cap V_c|] \ge (1-\eps) \ell_c$
        \item with a high-probability tail bound on the total fairness violation:
        \\
        for $\delta > 0$,
        $\bP[\fv(S) > (1+\delta) \eps \sum_c \ell_c] \le \exp\rb{-\Omega_\delta\rb{\sum_c \ell_c}}$
    \end{itemize}
    where $N$ is the maximum size of a set in $\cI \cap \cU$,
    and $\OPTMatInt$ is the maximum $f$-value of a set in $\cI \cap \cU$ (clearly we have $\OPTMatInt \ge \OPT$ as $\cC \subseteq \cU$).
\end{theorem}

We stress that $S$ is upper-fair with probability 1, not only in expectation.
We also remark that one can show a similar tail bound for every individual $\ell_c$, though the right-hand side $\exp(-\Omega_\delta(\ell_c))$ may not be meaningful unless $\ell_c$ is large.
On the other hand, no such bound can be shown for the $f$-value, which
in the worst case
can be concentrated on a single element of the universe.

The guarantee $\bE[f(S)] \ge 0.499 \cdot \eps \cdot \OPTMatInt$ of the second bullet point comes from using the local search algorithm of \Cref{thm:local_search} as a subroutine. We can instead use the simpler algorithm of \Cref{thm:mi-greedy} to get a slightly worse guarantee of $\bE[f(S)] \ge \frac{1}{3} \cdot \eps \cdot \OPTMatInt$;
we do so in our experimental evaluation.

\iffull
\begin{proofof}{\cref{thm:main-randomized}}
\else
We give a brief sketch; the full proof can be found \infullversion.
\begin{proofsketchof}{\cref{thm:main-randomized}}
\fi
    As the first step, we compute a maximum-cardinality fair and independent set~$P$, which may be done in polynomial time by \cref{lem:polytime-linear-fmmsm}.
\iffull
    We can say that $|P|=N$,
    i.e.,
    the maximum size of an independent and fair set
    is the same as the maximum size of an independent and upper-fair set.
    To see this, suppose that there was an independent and upper-fair set $F$ with $|F| > |P|$;
    then we could apply \cref{lem:exist-k-paths} to $P$ and $F$ to obtain an augmenting set $X$, and $P \triangle X$ would be a larger independent and fair set, a contradiction.
    
\fi
    As the second step, we compute a high-value independent and upper-fair set $Y_0$.
    Using the algorithm of \cref{thm:local_search} (\cite{Lee2010}) (with $\delta = 10^{-3}$) we get that
\iffull
    \begin{equation} \label{eq:fY0}
        f(Y_0) \ge 0.499 \cdot \OPTMatInt \,.
    \end{equation}
\fi
    We denote
    \[ k(Y) = \sum_{c \in [C]} \max(0, |P \cap V_c| - |Y \cap V_c|) \]
    for any solution $Y$,
    and $k := k(Y_0)$ to shorten notation.
\iffull

\fi
    We will perform a number $I$ of iterations
    which will be $(1-\eps)k$ in expectation.
    More precisely, let us set
    $I = \ceil{(1-\eps)k}$ with probability $(1-\eps)k-\floor{(1-\eps)k}$, and $\floor{(1-\eps)k}$
\iffull
    otherwise.\footnote{Ideally we would just set $I = (1-\eps)k$, but this number can be fractional, and using a fixed value of $\floor{(1-\eps)k}$ or $\ceil{(1-\eps)k}$ would lead to losses in objective value, cardinality, or fairness.
    For example, if $\ell_c = 1$, then a bound such as
    $|S \cap V_c| \ge (1-\eps)\ell_c - 1$ would be meaningless.
    }
\fi

    We perform $I$ iterations.
    In the $i$-th iteration,
    we apply \cref{lem:exist-k-paths} to $Y_{i-1}$ (and $P$)
    to obtain a collection $X_i^1,...,X_i^{k(Y_{i-1})}$
    of augmenting or alternating sets.
    We choose one of them,
    \iffull
    $X_i \in \{X_i^1,...,X_i^{k(Y_{i-1})}\}$,
    \fi
    uniformly at random,
    and apply it to obtain a new solution $Y_i = Y_{i-1} \triangle X_i$.
    Finally, we return $S := Y_I$.

\iffull
    All solutions $Y_0, ..., Y_I$ are independent and upper-fair; it remains to verify the guarantees of \cref{thm:main-randomized}.
    We start by noting that
    \begin{equation}
        k(Y_i) = k - i \,. \label{eq:k}
    \end{equation}
    To see this, note that during the algorithm's execution,
    no new color ever becomes undersaturated,
    as
    by \cref{def:paths}, $Y_i$ can have fewer elements than $Y_{i-1}$ in a color $c''$ only if $c''$ was oversaturated in $Y_{i-1}$.
    On the other hand, for exactly one undersaturated color $c'$,
    $Y_i$ has one more element in $c'$ than $Y_{i-1}$.
    Thus we have $k(Y_i) = k(Y_{i-1})-1$ and \eqref{eq:k} follows.
    (Colors $c$ that are neither under- or oversaturated remain such forever.)

    \paragraph{Fairness lower bounds.}
    Building upon the previous paragraph,
    we consider a random process involving colored balls
    that will mirror what is happening in the algorithm.
    Let $U \subseteq [C]$ be the set of colors that are undersaturated in $Y_0$.
    At the beginning, for every $c \in U$,
    we create $|P \cap V_c| - |Y_0 \cap V_c|$ balls of color $c$.
    (So we start with $k$ balls in total.)
    At every iteration $i$
    there is exactly one color $c'$
    (that is undersaturated in $Y_{i-1}$, so $c' \in U$)
    where $|Y_i \cap V_{c'}| = |Y_{i-1} \cap V_{c'}| + 1$;
    we then remove one random ball of color $c'$.
    Then,
    by \cref{def:paths} (since all other colors in $U$ retain their element count),
    we have that the number of balls of every color $c \in U$ is equal to $|P \cap V_c| - |Y_i \cap V_c|$
    (and their total number is $k(Y_i)=k-i$).

    Now we claim that in this process,
    at every iteration a uniformly random ball is removed.
    This is because, by \cref{lem:exist-k-paths},
    for every $c \in U$,
    exactly $|P \cap V_c| - |Y_{i-1} \cap V_c|$ of the $k(Y_{i-1})$ augmenting or alternating sets increase $c$,
    and we choose randomly among these sets.

    It follows that at the end,
    the set of removed $I$ balls
    is distributed uniformly among all subsets of this size.
    Consider a color $c$.
    If $c \not \in U$, then $c$ will not be undersaturated at the end, so $|S \cap V_c| \ge |P \cap V_c| \ge \ell_c$.
    Now fix $c \in U$
    and denote by $B_c$ the number of removed balls of color $c$.
    Conditioning on $I$,
    we have
    \begin{align*}
        \bE[|S \cap V_c|] &= \bE[|Y_0 \cap V_c| + B_c] \\
        &= |Y_0 \cap V_c| +  \frac{I}{k}(|P \cap V_c| - |Y_0 \cap V_c|) \\
        &\ge \frac{I}{k} |P \cap V_c| \\
        &\ge \frac{I}{k} \ell_c
    \end{align*}
    and thus
    $\bE[|S \cap V_c|] = \bE[\bE[|S \cap V_c| \mid I]] \ge \frac{\bE[I]}{k} \ell_c = (1-\eps) \ell_c$ as required.

    \paragraph{Cardinality.}
    Our proof that $\bE[|S|] \ge (1-\eps)|P| = (1-\eps)N$
    is very similar to the proof above for a single color.
    We start with $|P|-|Y_0|$ red balls and $k-(|P|-|Y_0|)$ non-red balls ($k$ in total).
    At every iteration $i$, if an augmenting set was chosen (so that $|Y_i| = |Y_{i-1}|+1$), we remove a red ball, otherwise we remove a non-red ball.
    Suppose that at every iteration $i$ there were exactly $|P|-|Y_{i-1}|$ augmenting sets among the $k-i+1$ sets;
    then the set of balls removed at the end would be distributed uniformly among all subsets of $I$ balls.
    Then, if $B$ denotes the number of removed red balls,
    we would have $\bE[B] = \frac{I}{k}(|P|-|Y_0|)$
    (conditioned on $I$).
    Now, in fact
    at every iteration $i$ there are \emph{at least} $|P|-|Y_{i-1}|$ augmenting sets among the $k-i+1$ sets;
    hence,
    the distribution of $B$
    dominates the above uniform-ball-subset distribution,
    which is formally known as $\Hypergeometric(k, |P|-|Y_0|, I)$.
    In particular,
    this implies that $\bE[B] \ge \frac{I}{k}(|P|-|Y_0|)$.
    We conclude by saying that
    $\bE[|S|] = |Y_0| + \bE[B] \ge |Y_0| + \frac{\bE[I]}{k}(|P|-|Y_0|) \ge (1-\eps)|P|$.

    \paragraph{Cardinality tail bound.}
    Recall that $N=|P|$,
    and that
    for any $\delta>0$ we want to prove that
        $\bP[|S| < (1-\delta) (1-\eps) N] \le \exp(-\Omega_\delta(N))$.
    It is known~\cite{hoeffding1963} that the hypergeometric distribution
    satisfies the same Chernoff-type bounds as the binomial distribution
    (as it corresponds to a sum of samples that are negatively correlated, rather than independent).
    In particular (conditioning on $I$ throughout), we have
    \begin{align*}
        \bP\sqb{B < \rb{1-\frac{\delta}{2}} \mu} &\le \exp\rb{-\frac12 \rb{\frac{\delta}{2}}^2 \mu} \le \exp\rb{-\Omega_\delta(\mu)} \\
        &\text{where } \mu = \bE[\Hypergeometric(k,|P|-|Y_0|,I)] = \frac{I}{k}(|P|-|Y_0|).
    \end{align*}
    If $|Y_0| \ge (1-\delta)(1-\eps)N$,
    then $|S| = |Y_0| + B$ is large enough with probability 1, so we can assume otherwise,
    i.e., that
    $|Y_0| < (1-\delta)(1-\eps)N \le (1-\delta)N$.
    Thus
    \begin{equation} \label{eq:k-is-large}
        k \ge |P|-|Y_0| \ge \delta N > \Omega(1),
    \end{equation}
    so for $N$ large enough we have $\frac{1}{k} < \frac{\delta}{2}(1-\eps)$ and thus
    $\frac{I}{k} \ge \frac{\floor{(1-\eps)k}}{k} \ge \frac{(1-\eps)k-1}{k} = 1-\eps-\frac{1}{k} \ge \rb{1-\frac{\delta}{2}}(1-\eps)$.
    By this and \eqref{eq:k-is-large},
    $\mu = \frac{I}{k}(|P|-|Y_0|) \ge \rb{1-\frac{\delta}{2}}(1-\eps)\delta N \ge \Omega_\delta(N)$,
    so that $\exp\rb{-\Omega_\delta(\mu)} = \exp\rb{-\Omega_\delta(N)}$.
    Finally, if the good event $B \ge \rb{1-\frac{\delta}{2}} \mu$ happens, then
    \[
    B \ge \rb{1-\frac{\delta}{2}} \frac{I}{k} (|P|-|Y_0|)
    \ge \rb{1-\frac{\delta}{2}}^2 (1-\eps) (|P|-|Y_0|)
    \ge \rb{1-\delta} (1-\eps) (|P|-|Y_0|)
    \]
    and thus
    $|S| = |Y_0| + B \ge (1-\delta)(1-\eps)N$.

    \paragraph{Total fairness violation tail bound.}
    Let us first remark that
    the algorithm increases some undersaturated color $c'$ at every iteration,
    so one could think that $\fv(S)$ is small with probability 1.
    However, undersaturation is measured with respect to $P$,
    and we can have $|P \cap V_c| \gg \ell_c$ for some $c$.
    Increasing a color beyond $\ell_c$ elements does not make progress with respect to fairness violation.
    Nevertheless, we can prove a high-concentration bound in terms of $\ell_c$.    
    Recall that
    for any $\delta > 0$
    we want to show that
        $\bP[\fv(S) > (1+\delta) \eps \sum_c \ell_c] \le \exp\rb{-\Omega_\delta\rb{\sum_c \ell_c}}$.
        
    The proof will be similar as above,
    but now the balls, on top of having a color,
    can be \emph{striped} or not.
    Namely, for each $c \in [C]$, we create
    $\max(0, |P \cap V_c| - |Y_0 \cap V_c|)$ balls of color $c$,
    of which
    $\max(0, \ell_c - |Y_0 \cap V_c|)$ many will be \emph{striped}.
    (We have $k$ balls in total, of which $\fv(Y_0)$ are striped.)
    Again, at each iteration, if $c'$ is the color that the algorithm increases,
    we remove a random ball of color $c'$.

    Let $X$ be the number of striped balls removed by the end.
    We then have
    \begin{equation}
        \label{eq:fvS}
        \fv(S) \le \fv(Y_0) - X \,.
    \end{equation}
    This is because whenever we increase some color $c'$ that has fewer than $\ell_{c'}$ elements,
    the fairness violation of the solution decreases by 1,
    but $X$ only accounts for this decrease if we happen to sample a \emph{striped} $c'$-colored ball.
    Since there are only as many striped $c'$-colored balls as there are
    fairness violations of color $c'$,
    at the end we will have removed at least as many of the violations as of the balls.
    
    Now we proceed as for cardinality.
    We have
    \begin{align*}
        \bP\sqb{X < \rb{1-\frac{\delta \eps}{2}} \mu} &\le \exp\rb{-\frac12 \rb{\frac{\delta \eps}{2}}^2 \mu} \le \exp\rb{-\Omega_\delta(\mu)} \\
        &\text{where } \mu = \bE[\Hypergeometric(k,\fv(Y_0),I)] = \frac{I}{k}\fv(Y_0).
    \end{align*}
    If $\fv(Y_0) \le (1+\delta)\eps \sum_c \ell_c$ then we are done,
    so assume otherwise. Then
    \begin{equation} \label{eq:k-is-large-2}
        k \ge \fv(Y_0) > (1+\delta)\eps \sum_c \ell_c > \Omega(1),
    \end{equation}
    so for $\sum_c \ell_c$ large enough we have
    $\frac{1}{k} < \frac{\delta \eps}{2} (1-\eps)$
    and thus
    $\frac{I}{k} \ge 1-\eps-\frac{1}{k} \ge \rb{1-\frac{\delta \eps}{2}}(1-\eps)$.
    By this and \eqref{eq:k-is-large-2},
    $\mu = \frac{I}{k}\fv(Y_0) \ge \rb{1-\frac{\delta \eps}{2}}(1-\eps) (1+\delta)\eps \sum_c \ell_c \ge \Omega(\sum_c \ell_c)$,
    so that
    $\exp(-\Omega_\delta(\mu)) = \exp(-\Omega_\delta(\sum_c \ell_c))$.
    Finally, if the good event $X \ge \rb{1-\frac{\delta \eps}{2}} \mu$ happens, then
    \[
    X \ge \rb{1-\frac{\delta \eps}{2}} \frac{I}{k} \fv(Y_0)
    \ge \rb{1-\frac{\delta \eps}{2}}^2 (1-\eps) \fv(Y_0)
    \ge \rb{1-\delta \eps - \eps} \fv(Y_0)
    \]
    and thus
    \[
    \fv(S) \stackrel{\eqref{eq:fvS}}{\le} \fv(Y_0) - X
    \le \fv(Y_0) - \rb{1-\delta \eps - \eps} \fv(Y_0)
    = (1 + \delta) \eps \fv(Y_0)
    \le (1+\delta) \eps \sum_c \ell_c \,.
    \]

    \paragraph{Objective value.}
    Intuitively, at every iteration $i$,
    we select randomly from among $k-i+1$ disjoint augmenting or alternating sets.
    Even if the newly added elements do not add any $f$-value,
    by submodularity
    \iffull\else (\cref{fact:submod_k}) \fi
    we expect to lose only at most a $1/(k-i+1)$ fraction of the $f$-value of the current solution.
    After $I \approx (1-\eps)k$ iterations we then end up with a telescoping product that simplifies to $\frac{\eps k}{k} f(Y_0)$.

    We now give a formal proof.
    We show by induction on $i$ that 
    \begin{equation} \label{eq:fYi}
        \bE[f(Y_i)] \ge \frac{k-i}{k} f(Y_0) \,.
    \end{equation}
    For $i \ge 1$, condition on $Y_{i-1}$. Then
    \begin{align*}
        \bE[f(Y_i)] &= \bE[f(Y_{i-1} \triangle X_i)] \\
        &= \frac{1}{k(Y_{i-1})} \sum_{j=1}^{k(Y_{i-1})} f(Y_{i-1} \triangle X_i^j) \\
        &\ge \frac{1}{k-i+1} \sum_{j=1}^{k-i+1} f(Y_{i-1} \setminus (Y_{i-1} \cap X_i^j)) \\
        &\ge \frac{k-i}{k-i+1} f(Y_{i-1}) \,,
    \end{align*}
    where the first inequality follows by monotonicity
    and the second inequality is by applying \cref{fact:submod_k} to the set family $Y_{i-1} \cap X_i^1, ..., Y_{i-1} \cap X_i^{k-i+1} \subseteq Y_{i-1}$.
    Now taking expectation over $Y_{i-1}$,
    \begin{align*}
        \bE[f(Y_i)] &= \bE[\bE[f(Y_i) \mid Y_{i-1}]] \\
        &\ge \bE\left[ \frac{k-i}{k-i+1} f(Y_{i-1}) \right] \\
        &\ge \frac{k-i}{k-1+1} \cdot \frac{k-(i-1)}{k} f(Y_0)
    \end{align*}
    where we applied the inductive hypothesis.
    Having \eqref{eq:fY0} and \eqref{eq:fYi}, we can write
    \[ \bE[f(S)] = \bE[\bE[f(Y_I) \mid I]] \ge \bE\left[\frac{k-I}{k} f(Y_0)\right] = \frac{k-(1-\eps)k}{k} f(Y_0) \ge \eps \cdot 0.499 \cdot \OPTMatInt \,. \]
\end{proofof}
\else
    All solutions $Y_0, ..., Y_I$
    are independent and upper-fair.
    The properties of fairness lower bounds intuitively follow because at every iteration one random fairness violation is removed,
    and the number of iterations is $\approx (1-\eps)$ times the initial number of fairness violations $k=k(Y_0)$.
    Since at least $|P|-|Y_{i-1}|$ of the sets at iteration $i$ are augmenting, the claim about the size of the solution follows similarly.
    We can also show tail bounds by invoking Chernoff-Hoeffding style bounds for hypergeometric distributions.
    As for the objective value, we prove that at every step, the expected loss is only a $1/(k-i+1)$ fraction of the current $f$-value, as we select randomly from among $k-i+1$ disjoint augmenting or alternating sets. After $I \approx (1-\eps)k$ iterations we then end up with a telescoping product that simplifies to $\frac{\eps k}{k}f(Y_0)$.
\end{proofsketchof}
\vspace{-0.5cm}
\fi
\iffull

The pseudocode for \Cref{lem:exist-k-paths} can be found in \Cref{alg:lem:exist-k-path} and the pseudocode for \Cref{thm:main-randomized} can be found in \Cref{alg:main-randomized}.

\begin{algorithm}
\caption{Generate Augmenting Paths}
\label{alg:lem:exist-k-path}
\begin{algorithmic}[1]
\REQUIRE Matroid $\mathcal{I}$ on universe $V$ and fairness constraints with colors $[C]$ and lower bound $\ell_c$ and upper bound $u_c$ for each color $c \in [C]$.
\REQUIRE Sets $Y$ and $P$, where $Y$ and $P$ are independent sets with respect to matroids $\mathcal{I}$ and $\mathcal{U}$, and $|Y| \leq |P|$.
\ENSURE A set of augmenting and alternating paths between $Y$ and $P$. 
\STATE Build an initially empty bipartite graph with ``left vertices" being elements of $Y$ and ``right vertices" being elements of $P$
\FOR{each $(y, p) \in Y \times P$}
\STATE Add directed edge $ y \to p$ if $Y + p - y \in \cU$
\STATE Add directed edge $y \leftarrow p$ if $Y + p - y \in \cI$
\ENDFOR
\STATE Define $P$-sinks $T_P := \{ p \in P : Y + p \in \cI \}$ and let $T_P'$ be an arbitrary subset of $T_P$ of size exactly $|P|-|Y|.$
\STATE Find a matching between $Y$ and $P \setminus T_P'$ and drop edges adjacent to $T_P \setminus T_P'$. Call this $M_{\leftarrow}$.
\STATE Start with $M_{\rightarrow} = \emptyset$.
\FOR{every color $c \in [C]$}
\STATE Add $\min(|Y \cap V_c|, |P \cap V_c|)$ edges from $Y \cap V_c$ to $P \cap V_c$ to the matching $M_\to$.
\ENDFOR
\STATE Define the set $S$ of \emph{sources} as all vertices in $P$ that did not get matched in $M_\to$.
\STATE Define the set $T_Y$ of \emph{$Y$-sinks} as all vertices in $Y$ that did not get matched in $M_\to$.
\STATE Let $\mathcal{X}$ be an initially empty collection of paths. 
\WHILE{there is an unused source}
\STATE Start at at unused source (in $P$).
\IF{at a vertex of $P$}
\IF{vertex is a sink (in $T_P$)}
\STATE Stop path and add to collection $\mathcal{X}$.
\ELSE
\STATE Follow the outgoing edge of $M_\leftarrow$.
\ENDIF
\ELSE 
\IF{vertex is a sink (in $T_Y$)}
\STATE Stop path and add to collection $\mathcal{X}$.
\ELSE
\STATE Follow the outgoing edge of $M_\rightarrow$.
\ENDIF
\ENDIF
\ENDWHILE

\FOR{each path $X$ in $\mathcal{X}$}
\STATE Shortcut $X$ along right-to-left directed edges (see proof of \Cref{lem:exist-k-paths} for details).  
\ENDFOR

\RETURN collection of disjoint augmenting / alternating sets $\mathcal{X}$.

\end{algorithmic}
\end{algorithm}

\begin{algorithm}
\caption{Random Augmenting Paths}
\label{alg:main-randomized}
\begin{algorithmic}[1]
\REQUIRE Matroid $\mathcal{I}$ on universe $V$, a submodular function $f: 2^V \rightarrow \mathbb{R}_{\geq 0}$, fairness constraints with colors $[C]$ and lower bound $\ell_c$ and upper bound $u_c$ for each color $c \in [C]$.
\ENSURE Approximately fair and high value set $S$.
\STATE Find a fair independent set $P$ via \Cref{lem:polytime-linear-fmmsm} or any matroid intersection algorithm. 
\STATE Find a high $f$-value independent set $Y_0$ with respect to matroids $\mathcal{I}$ and $\mathcal{U}$ via either local search (to get guarantee \Cref{thm:local_search}) or Greedy (to get guarantee \Cref{thm:mi-greedy}). 
\STATE Let $k(Y)$ be the number of paths returned by imputing the matchings $Y$ and $P$ into \Cref{alg:lem:exist-k-path}.
\STATE Let $I$ be $\lfloor (1-\varepsilon) k(Y_0) \rfloor$ with probability $\lceil  (1-\varepsilon) k(Y_0) \rceil - (1-\varepsilon) k(Y_0)$ and $\lceil  (1-\varepsilon) k(Y_0) \rceil$ otherwise.
\FOR{$i = 1$ to $I$}
    \STATE Compute feasible augmenting/alternating paths $X_{i}^1, \hdots, X_{i}^{k(Y_{i-1})}$ between $Y_{i-1}$ and $P$ as outlined in \Cref{lem:exist-k-paths} (see \Cref{alg:lem:exist-k-path}).
    \STATE Choose one $X_i \in \{X_i^1,...,X_i^{k(Y_{i-1})}\}$ uniformly at random. 
    \STATE Apply it to obtain $Y_i = Y_{i - 1} \Delta X_i$.
\ENDFOR
\RETURN $S := Y_{I}$ 
\end{algorithmic}
\end{algorithm}

\fi

\subsection{Deterministic algorithm}
\label{sec:deterministic}

\iffull
Now we turn to our deterministic result, \cref{thm:main-deterministic}.
We begin by showing a lemma that is an analogue of \cref{lem:exist-k-paths}.
\fi
\begin{lemma} \label{lemma:exist-k-paths-det}
    For any two matroids $\cI_1$, $\cI_2$,
    let $Y, P \in \cI_1 \cap \cI_2$ be two sets in their intersection,
    with $|Y| \le |P|$.
    Then we may find in polynomial time a collection $X_1, ..., X_{|P|-|Y|}$ of disjoint subsets of $Y \cup P$
    such that for each set $X_i$ we have $Y \triangle X_i \in \cI_1 \cap \cI_2$ and $|Y \triangle X_i| = |Y| + 1$.
\end{lemma}
\iffull
\begin{proof}
    We proceed similarly as in the proof of \cref{lem:exist-k-paths}.
    We consider the matroid intersection exchange graph for $Y$ and $P$ with respect to $\cI_1$ and $\cI_2$, defined as in that proof.
    Define $T = \{ p \in P : Y+p \in \cI_1\}$
    and $S = \{ p \in P : Y+p \in \cI_2\}$.
    We have $|T|, |S| \ge |P|-|Y|$.
    Let $T'$, $S'$ be arbitrary subsets of $T$, $S$
    respectively, both of size exactly $|P|-|Y|$.
    We invoke \cref{lem:perfect-matching} on the exchange graph for $Y$ and $P \setminus T'$, obtaining $M_\leftarrow$
    as a perfect matching between $Y$ and $P \setminus T'$.
    Similarly, we obtain $M_\to$
    as a perfect matching between $Y$ and $P \setminus S'$.

    Now we can construct the $|P|-|Y|$ paths as in the proof of \cref{lem:exist-k-paths}, starting from sources (set $S'$) and proceeding in the only possible way until we reach a sink (vertex in $T'$).
    (Alternatively, we can note that $M_\leftarrow \cup M_\to$ is a circulation for demands -- i.e., outflow minus inflow -- $+1$ on sources and $-1$ on sinks, and take the paths from its cycle-path decomposition.)
    All paths start and end in $P$.

    Next, we shortcut the paths.
    Here, we replace subpaths with chords
    in both directions, rather than only in one direction
    as in the proof of \cref{lem:exist-k-paths}.
    Moreover, if there is an internal vertex that is in $S$, we need to truncate the path so that it begins at that vertex;
    and similarly, if there is an internal vertex that is in $T$,
    we truncate the path so that it ends at that vertex.
    These operations only shrink the vertex sets of the paths,
    thus they preserve their vertex-disjointness.
    
    We end once the path $X_i$ is from $S$ to $T$ via $(Y \cup P) \setminus (S \cup T)$ and contains no chords.
    The same argument as in the proof of \cref{lem:exist-k-paths} then shows that $X_i$ is an augmenting path,
    i.e., that $Y \triangle X_i \in \cI_1 \cap \cI_2$.
\end{proof}

Now we can state our deterministic algorithm for any two matroids $\cI_1$ and $\cI_2$.

\begin{theorem} \label{thm:main-deterministic}
    There is a deterministic polynomial-time algorithm for
    the problem of maximizing a monotone submodular function
    subject to a matroid intersection constraint,
    parametrized by $\eps \in (0,1)$,
    that outputs a set $S \in \cI_1 \cap \cI_2$
    such that
    \begin{itemize}
        \item $|S| > (1-\eps)N - 1$
        \item $f(S) \ge 0.499 \cdot \eps \cdot \OPTMatInt$
    \end{itemize}
    where $N$ is the maximum size of a set in $\cI_1 \cap \cI_2$,
    and $\OPTMatInt$ is the maximum $f$-value of a set in $\cI_1 \cap \cI_2$.
\end{theorem}
\begin{proof}
    As in the algorithm of \cref{thm:main-randomized},
    we start by computing a maximum-cardinality set $P$ in the matroid intersection,
    as well as a high-value set $Y$ in the matroid intersection.
    We have $|P| = N$ and $f(Y_0) \ge 0.499 \cdot \OPTMatInt$.

    We then perform $I := \floor{(1-\eps)(|P|-|Y_0|)}$ iterations.
    In the $i$-th iteration,
    we apply \cref{lemma:exist-k-paths-det} to $Y_{i-1}$ (and $P$)
    to obtain a collection $X_i^1,...,X_i^{|P|-|Y_{i-1}|}$
    of sets.
    We choose the one of them, $X_i \in \{X_i^1,...,X_i^{|P|-|Y_{i-1}|}\}$, that,
    when used to obtain a new solution $Y_i = Y_{i-1} \triangle X_i$,
    maximizes $f(Y_i)$.
    Finally, we return $S := Y_I$.

    All solutions $Y_0, ..., Y_I$ are in the matroid intersection,
    and they grow in size by 1 per step.
    Thus we have
    $|S| = |Y_0| + \floor{(1-\eps)(|P| - |Y_0|)} > |Y_0| + (1-\eps)(|P| - |Y_0|) - 1 \ge (1-\eps)|P| - 1 = (1-\eps)N - 1$.
    
    The proof for the objective value guarantee is the same as in \cref{thm:main-randomized};
    at each step, since an average set preserves a $\frac{|P|-|Y_0|-i}{|P|-|Y_0|-i+1}$ fraction of the value, so does the best set.
    We conclude by saying that since $I \le (1-\eps)(|P|-|Y_0|)$,
    \[f(S) = f(Y_I) 
    \ge \frac{|P|-|Y_0|-(1-\eps)(|P|-|Y_0|)}{|P|-|Y_0|} f(Y_0) \ge \eps \cdot 0.499 \cdot \OPTMatInt \,.
    \]
\end{proof}
\fi

\section{Experimental evaluation}
\label{sec:experiments}

\iffull
\begin{figure}[h!]
  \centering
  \includegraphics[width=\textwidth]{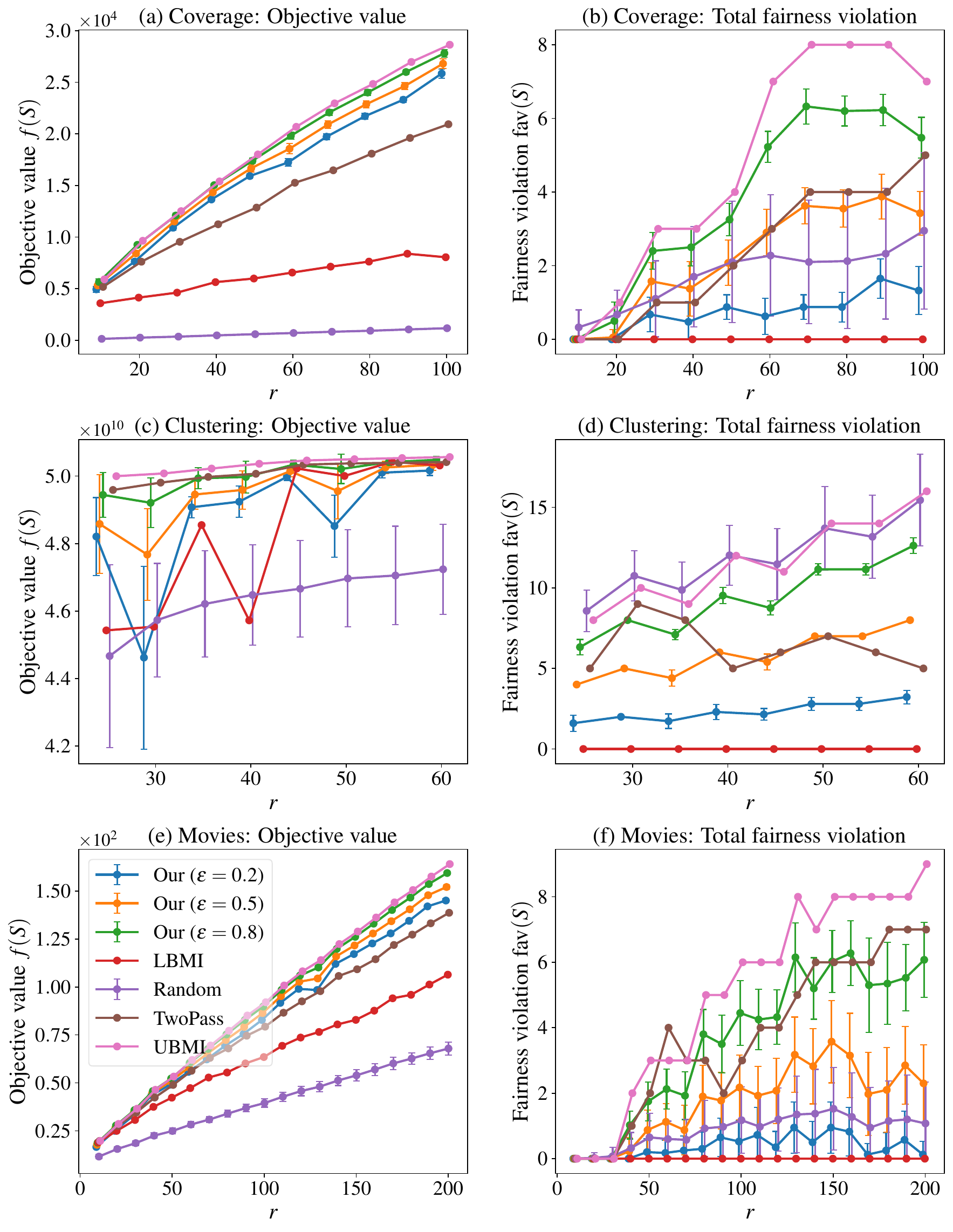}
  \caption{Our experimental results. Each row corresponds to one experiment; the left plot shows the objective value of each algorithm for a range of solution scale factors $r$, and the right plot shows fairness violations. For randomized algorithms we report averages, with error bars that correspond to sample standard deviation.}
  \label{fig:results}
\end{figure}
\begin{figure}[h!]
  \centering
  \includegraphics[width=\textwidth]{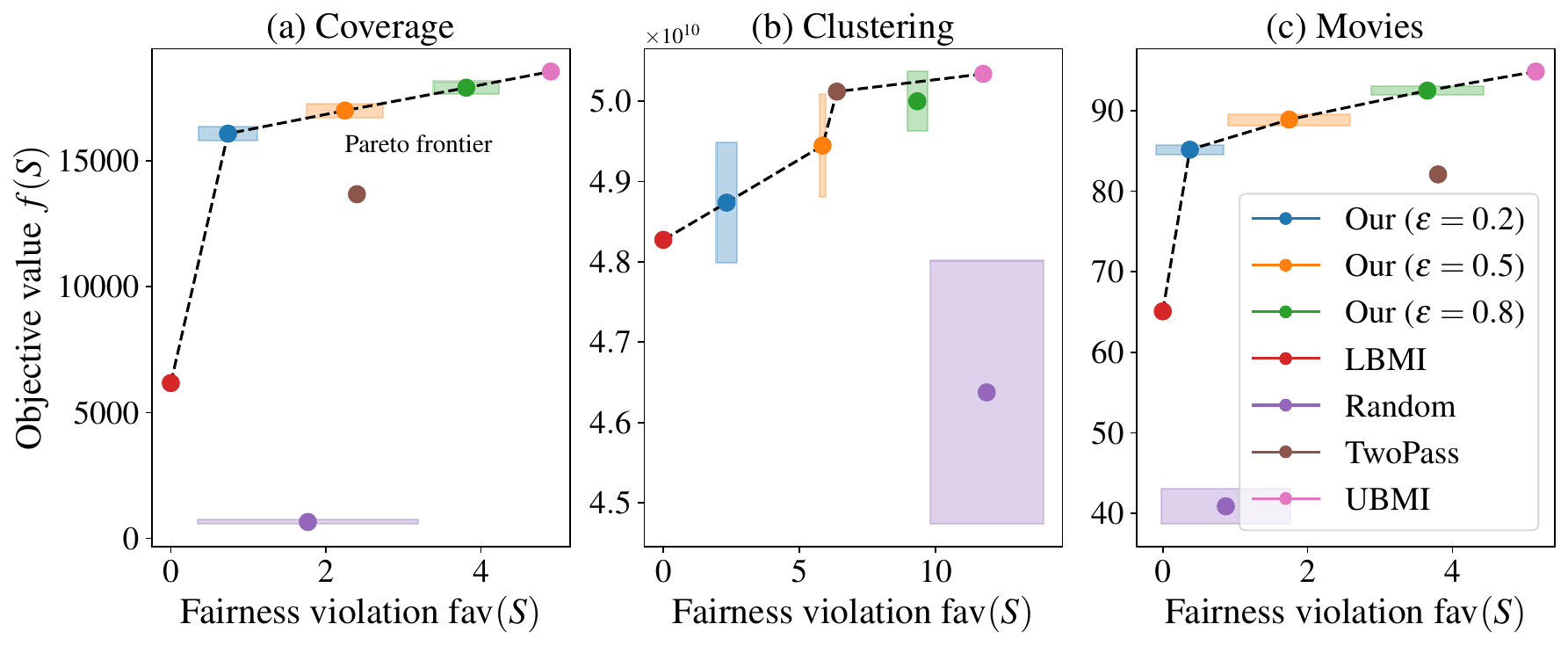}
  \caption{
    For each experiment and algorithm we take the average objective value and fairness violation over all $r$-values, and plot this as a single point. For randomized algorithms, 
    the colored rectangles correspond to standard deviations.
    The dashed line corresponds to the Pareto frontier of the trade-off between objective value and fairness violation.
  }
  \label{fig:scatter}
\end{figure}
\fi

We evaluate the performance of our algorithms
empirically
against prior work and natural baselines
closely following the experimental setup of prior work~\cite{HalabiMNTT20,el2023fairness},
on a suite of benchmarks
that are standard in the field:
graph coverage, clustering, and recommender systems,
under different fairness and matroid constraint settings.
Our metrics are the submodular objective value $f(S)$
and total fairness violation $\fv(S)$.
All of the considered algorithms return sets that are independent and upper-fair, so the measured fairness violations are all with respect to the lower bounds. LLMs were used to assist in coding. 
\iffull
All three benchmarks use a partition matroid.  

\fi
We compare the following algorithms:
\newcommand{\LBMI}{\textsc{LBMI}\xspace}
\newcommand{\UBMI}{\textsc{UBMI}\xspace}
\newcommand{\Our}{\textsc{Our}\xspace}
\newcommand{\Random}{\textsc{Random}\xspace}
\newcommand{\TwoPass}{\textsc{TwoPass}\xspace}
\iffull
\begin{itemize}
\else
\begin{itemize}[topsep=0pt,itemsep=0pt,leftmargin=*]
\fi
    \item $\Our(\eps)$ -- our algorithm of \cref{thm:main-randomized}, for a range of settings of $\eps \in \{0.2,0.5,0.8\}$. To compute a high-value solution $Y$, we run the natural greedy algorithm, which obtains a $1/3$-approximation (\cref{thm:mi-greedy}),
    as the local search algorithm of \iffull \cref{thm:local_search}, while polynomial-time, is \fi impractical. The large fair set $P$ is obtained via augmenting paths, ignoring $f$.

    \item $\TwoPass$ -- the algorithm of \cite{el2023fairness} (\cref{thm:streaming}). Since it was originally developed for the streaming setting, to get a fair comparison we simplify away the parts (namely the first pass) whose purpose was ensuring low memory usage. The first step of the algorithm obtains a fair set via augmenting paths (ignoring $f$). This is then divided in two, and each half is extended to an independent and upper-fair solution using a matroid intersection subroutine. For this we employ the greedy algorithm (the original implementation of \cite{el2023fairness} used a swapping algorithm to ensure low memory and linear runtime, but it obtains inferior values).

    \item \LBMI (Lower Bound Matroid Intersection) -- an algorithm that always returns a fair set, with no theoretical guarantee on the
    \iffull
    value
    but with reasonably good value in practice (similar in spirit to \textsc{Greedy-Fair-Streaming} from~\cite{el2023fairness}).
    \fi
    It starts by building a fair set via augmenting paths, ignoring $f$, and then extends to a maximal solution using the greedy algorithm.

    \item \UBMI (Upper Bound Matroid Intersection) -- an algorithm that ignores lower bound constraints and just solves the matroid intersection problem for $\cI$ and $\cU$ (similar in spirit to \textsc{Matroid-Intersection} from~\cite{el2023fairness}). Also here we use the greedy algorithm.

    \item $\Random$ -- an algorithm that randomly shuffles the universe and then adds each element if this keeps the solution independent and upper-fair.

\end{itemize}

For a fair comparison of the main underlying ideas,
we made sure that the compared algorithms, particularly \Our and \TwoPass,
use the same subroutines for similar tasks;
the implementations could likely
benefit from heuristically taking $f$ into account rather than ignoring $f$ when building large fair sets,
or from some local-search based postprocessing of the final solution.
We do not compare to the algorithm of~\cite{el2024} (\cref{thm:inexp}) as solving the multilinear extension makes it impractical.
\iffull

\fi
We repeat the randomized algorithms 40 times.
All experiments can be run on commodity hardware (CPU only, single-threaded; we do not report runtimes) and take several hours to finish.

The code for the paper is available at \url{https://github.com/dj3500/fair-matroid-submodular-neurips2025}.


We outline the experimental scenarios below.
In each experiment we vary a solution size scaling factor $r$, which roughly corresponds to the rank of the matroid $\cI$.
\iffull
We select fairness bounds $\ell_c, u_c$ to ensure feasibility, requiring that each color group $V_c$ is proportionally represented in the solution set $S$ -- either matching its share in the dataset (coverage, movies) or ensuring similar group sizes (clustering).
Results are reported in \cref{fig:results,fig:scatter} and discussed in \cref{sec:results}.
\fi

\paragraph{Computational complexity.}
We start with the complexity of the general randomized algorithm of \Cref{thm:main-randomized}. Firstly, the runtime of constructing $P$ (a maximum-cardinality fair and independent set) via augmenting paths is $O(N^{1.5}|V|)$ (by \cite{Schrijver}, Chapter 41.2 Notes). To construct $Y$ (an upper-fair and independent set of high $f$-value), we expend $O(N|V|)$ time using the greedy algorithm.

Next, at each of the $I$ iterations, we must (1) recompute the exchange graph between $Y_i$ and $P$, (2) find $M_{\leftarrow}$ and $M_{\rightarrow}$ as the subgraph of interest, (3) decompose $M_{\leftarrow} \cup M_{\rightarrow}$ into paths, and (4) shortcut these paths. Step (1) takes $O(N^2)$ time, since we query if a directed edge exists between $y$ and $p$ for all $y \in Y$ and $p \in P$. Finding perfect matchings in step (2) takes at most $O(N^{3})$ time (in a practical implementation we could use the Hopcroft-Karp algorithm). Decomposing the resulting subgraph into paths takes at most $O(N)$ time. And lastly, shortcutting the paths again takes at most $O(N^2)$ time. Since $I$ can be $\Theta(N)$, the total runtime is at most $O(N^{4})$. 

A more efficient implementation is possible if $\cI$ is a partition matroid.
The intersection of two partition matroids can be naturally interpreted as a bipartite multigraph (the colors, i.e., parts of $\mathcal{U}$ are one side, the parts of $\mathcal{I}$ are the other side, and an element corresponds to an edge between the two parts it belongs to). In this case, we may look at the following exchange graph: direct the edges of $Y$ from left to right, and the edges of $P$ from right to left. This directed graph may be decomposed into paths. These paths are \textit{simultaneously feasible}, and so we do not need to recompute an exchange graph at every step (or shortcut). Since there are $O(N)$ edges, the runtime to decompose this directed graph is $O(N)$. Over the $I$ iterations, we have a total runtime of at most $O(N^2)$.



\iffull
\subsection{Graph coverage}
\fi
We use the Pokec social network~\citep{snapnets}.
Given a digraph $G = (V, E)$ of users and their friendships, we select a subset $S \subseteq V$ to maximize coverage, defined by $f(S) = \left|\bigcup_{v \in S} N(v)\right|$, where $N(v)$ is the neighborhood of $v$. User profiles include age, gender, height, and weight. We impose a partition matroid on body mass index (BMI). Profiles missing height or weight or with implausible data are removed, yielding a graph with 582,289 nodes and 5,834,695 edges. Users are partitioned into four BMI categories (underweight, normal, overweight, obese), with upper bounds $\ceil{ \frac{|V_i|}{|V|} r}$ for each group $V_i$. We also enforce fairness by age, with 7 groups: $[1,10], [11,17], [18,25], [26,35], [36,45], [46+],$ no age. We set $\ell_c = \floor{0.9 \frac{|V_c|}{|V|} r}$ and $u_c = \ceil{1.5 \frac{|V_c|}{|V|} r}$. We use $r$ from $10$ to $200$.

\iffull
\subsection{Exemplar-based clustering}
\fi
We use a dataset of 4521 phone calls from a Portuguese bank marketing campaign~\citep{MoroCR14}. The goal is to select a representative subset $S \subseteq V$ for service quality assessment. Each record $e \in V$ is represented as $x_e \in \bR^7$ using 7 numeric features, including age and account balance. We impose a partition matroid on account balance, with 5 groups: $(-\infty,0), [0,2000), [2000,4000), [4000,6000), [6000,\infty)$. Each group $V_i$ has upper bound $r/5$. Fairness is enforced by age, with 6 groups: $[0,29], [30,39], [40,49], [50,59], [60,69], [70+]$, and bounds $\ell_c = 0.1 r + 2$, $u_c = 0.4 r$ for each $c$. We maximize the monotone submodular function~\citep{krause2010budgeted}:
$
    f(S) = \sum_{e' \in V} \big(d(e',0) - \min_{e \in S \cup \{0\}} d(e',e) \big)
$
where $d(e',e) = \|x_{e'} - x_e\|_2^2$ and $x_0$ is the origin. We use $r$ from $30$ to $60$.

\iffull
\subsection{Recommender system}
\fi
We simulate a movie recommendation system using the Movielens 1M dataset~\citep{harper2016movielens}, with about one million ratings for 3900 movies by 6040 users. As in prior work~\citep{mitrovic2017streaming,norouzi2018beyond,HalabiMNTT20,el2023fairness}, we compute a low-rank completion of the user-movie matrix~\citep{troyanskaya2001missing}, yielding $w_u \in \bR^{20}$ for each user $u$ and $v_m \in \bR^{20}$ for each movie $m$. The product $w_u^\top v_m$ estimates user $u$'s rating for movie $m$. For user $u$, the monotone submodular utility for a set $S$ of movies is
$ f(S) =
\alpha \cdot \sum_{m' \in M} \max\rb{\max_{m \in S} \rb{v_m^\top v_{m'}}, 0} + (1-\alpha) \cdot \sum_{m \in S} w_u^\top v_m ,
$
with parameter $\alpha=0.85$ balancing coverage and personalized user score. We enforce proportional representation of movies by release date using a partition matroid with 9 decade groups (1911--2000), with upper bounds $\ceil{1.2 \frac{|V_d|}{|V|} r}$ for each decade $V_d$. Movies are also partitioned into 18 genres $c$ (colors), with fairness bounds $\ell_c = \floor{0.8 \frac{|V_c|}{|V|} r}$ and $u_c = \ceil{1.4 \frac{|V_c|}{|V|} r}$. We use $r$ from $10$ to $200$.

\subsection{Results and discussion}
\label{sec:results}

Our results are depicted in \cref{fig:results,fig:scatter}.
Similarly as prior work,
we observe that enforcing fairness does come at
some
cost in the utility value,
and that the utility values of the algorithms are much better in practice than the theoretical bounds guarantee.

In all three experiments, our algorithms produce solutions whose value is relatively competitive with \UBMI, which completely ignores the lower bound constraints
and accordingly has the highest fairness violations.
In two of the three scenarios (coverage and movies),
all \Our algorithms produce a higher $f$-value than all the other baselines (\Random, \LBMI, and \TwoPass);
in particular, \TwoPass is dominated by both $\Our(0.2)$ and $\Our(0.5)$ with respect to both metrics.
For clustering the situation is somewhat unclear, but \TwoPass generally does better.
In terms of violation of the lower bound fairness constraints, our different settings of $\eps$, as expected, provide a smooth tradeoff.
The baseline that guarantees no fairness violations, \LBMI, does relatively poorly in terms of $f$-value.

This tunability of $\eps$ is a key strength of our approach, allowing users to select an operating point that best matches their specific requirements for the balance between utility and fairness.

\section{Conclusion, limitations, broader impact, and future work}
\label{sec:conclusion}

In this work we gave an improved algorithm for FMMSM
which, for any $\eps > 0$,
returns an approximate solution
that satisfies an expected $(1-\eps)$ fraction of each fairness lower bound
while satisfying the matroid constraint and the fairness upper bound constraints exactly;
the returned solution is also large in size
and enjoys high-concentration guarantees.
\iffull

\fi
Recent studies have shown that automated algorithms used in decision-making can introduce bias and discrimination.
We make progress towards
mitigating such effects in problems that can be formulated as submodular maximization under a matroid constraint,
which are relevant to a range of applications
such as forming representative committees or curating content for news feeds.
We show the strong performance of our algorithm empirically on several real-world tasks.
As in prior work, we observe
\iffull that there is indeed \fi
a balance between fairness and utility value;
however, this ``price of fairness'' should not be interpreted as fairness leading to inferior outcomes,
but rather as a trade-off between two valuable metrics.
The parametric nature of our algorithm (the tunable $\eps$ parameter)
provides a new tool to help in navigating  this balance.

Our work leaves open the exciting question of the approximability of FMMSM
(without violations of fairness constraints)
and MSPM.
Is there a constant-factor approximation algorithm for MSPM?
Or is there a superconstant hardness of approximation for FMMSM?
(As remarked in~\cite{el2024},
the latter result would give a negative answer to a fundamental question posed by Vondrák~\cite{VondrakTalk}.)
We also do not consider non-monotone objective functions or the streaming setting in this work.
\iffull

\fi
Finally, it is important to note that the fairness notion we employ, though standard and general,
does not capture some notions of fairness considered in the literature (see e.g.~\cite{chouldechova2018frontiers, tsang2019}).
No universal definition of fairness exists;
the choice of which definition to apply
is application-dependent and an active area of research.

\printbibliography


\newpage
\section*{NeurIPS Paper Checklist}

\begin{enumerate}

\item {\bf Claims}
    \item[] Question: Do the main claims made in the abstract and introduction accurately reflect the paper's contributions and scope?
    \item[] Answer: \answerYes{} 
    \item[] Guidelines:
    \begin{itemize}
        \item The answer NA means that the abstract and introduction do not include the claims made in the paper.
        \item The abstract and/or introduction should clearly state the claims made, including the contributions made in the paper and important assumptions and limitations. A No or NA answer to this question will not be perceived well by the reviewers. 
        \item The claims made should match theoretical and experimental results, and reflect how much the results can be expected to generalize to other settings. 
        \item It is fine to include aspirational goals as motivation as long as it is clear that these goals are not attained by the paper. 
    \end{itemize}

\item {\bf Limitations}
    \item[] Question: Does the paper discuss the limitations of the work performed by the authors?
    \item[] Answer: \answerYes{} 
    \item[] Justification: Yes, particularly in \cref{sec:conclusion}.
    \item[] Guidelines:
    \begin{itemize}
        \item The answer NA means that the paper has no limitation while the answer No means that the paper has limitations, but those are not discussed in the paper. 
        \item The authors are encouraged to create a separate "Limitations" section in their paper.
        \item The paper should point out any strong assumptions and how robust the results are to violations of these assumptions (e.g., independence assumptions, noiseless settings, model well-specification, asymptotic approximations only holding locally). The authors should reflect on how these assumptions might be violated in practice and what the implications would be.
        \item The authors should reflect on the scope of the claims made, e.g., if the approach was only tested on a few datasets or with a few runs. In general, empirical results often depend on implicit assumptions, which should be articulated.
        \item The authors should reflect on the factors that influence the performance of the approach. For example, a facial recognition algorithm may perform poorly when image resolution is low or images are taken in low lighting. Or a speech-to-text system might not be used reliably to provide closed captions for online lectures because it fails to handle technical jargon.
        \item The authors should discuss the computational efficiency of the proposed algorithms and how they scale with dataset size.
        \item If applicable, the authors should discuss possible limitations of their approach to address problems of privacy and fairness.
        \item While the authors might fear that complete honesty about limitations might be used by reviewers as grounds for rejection, a worse outcome might be that reviewers discover limitations that aren't acknowledged in the paper. The authors should use their best judgment and recognize that individual actions in favor of transparency play an important role in developing norms that preserve the integrity of the community. Reviewers will be specifically instructed to not penalize honesty concerning limitations.
    \end{itemize}

\item {\bf Theory assumptions and proofs}
    \item[] Question: For each theoretical result, does the paper provide the full set of assumptions and a complete (and correct) proof?
    \item[] Answer: \answerYes{} 
    \item[] Guidelines:
    \begin{itemize}
        \item The answer NA means that the paper does not include theoretical results. 
        \item All the theorems, formulas, and proofs in the paper should be numbered and cross-referenced.
        \item All assumptions should be clearly stated or referenced in the statement of any theorems.
        \item The proofs can either appear in the main paper or the supplemental material, but if they appear in the supplemental material, the authors are encouraged to provide a short proof sketch to provide intuition. 
        \item Inversely, any informal proof provided in the core of the paper should be complemented by formal proofs provided in appendix or supplemental material.
        \item Theorems and Lemmas that the proof relies upon should be properly referenced. 
    \end{itemize}

    \item {\bf Experimental result reproducibility}
    \item[] Question: Does the paper fully disclose all the information needed to reproduce the main experimental results of the paper to the extent that it affects the main claims and/or conclusions of the paper (regardless of whether the code and data are provided or not)?
    \item[] Answer: \answerYes{} 
    \item[] Justification: \newcommand{\weprovidecode}{Yes, we provide the full source code in the supplementary material, zipped together with all relevant datasets for ease of execution. We will also open-source release the code together with the camera-ready version.}
    \weprovidecode
    \item[] Guidelines:
    \begin{itemize}
        \item The answer NA means that the paper does not include experiments.
        \item If the paper includes experiments, a No answer to this question will not be perceived well by the reviewers: Making the paper reproducible is important, regardless of whether the code and data are provided or not.
        \item If the contribution is a dataset and/or model, the authors should describe the steps taken to make their results reproducible or verifiable. 
        \item Depending on the contribution, reproducibility can be accomplished in various ways. For example, if the contribution is a novel architecture, describing the architecture fully might suffice, or if the contribution is a specific model and empirical evaluation, it may be necessary to either make it possible for others to replicate the model with the same dataset, or provide access to the model. In general. releasing code and data is often one good way to accomplish this, but reproducibility can also be provided via detailed instructions for how to replicate the results, access to a hosted model (e.g., in the case of a large language model), releasing of a model checkpoint, or other means that are appropriate to the research performed.
        \item While NeurIPS does not require releasing code, the conference does require all submissions to provide some reasonable avenue for reproducibility, which may depend on the nature of the contribution. For example
        \begin{enumerate}
            \item If the contribution is primarily a new algorithm, the paper should make it clear how to reproduce that algorithm.
            \item If the contribution is primarily a new model architecture, the paper should describe the architecture clearly and fully.
            \item If the contribution is a new model (e.g., a large language model), then there should either be a way to access this model for reproducing the results or a way to reproduce the model (e.g., with an open-source dataset or instructions for how to construct the dataset).
            \item We recognize that reproducibility may be tricky in some cases, in which case authors are welcome to describe the particular way they provide for reproducibility. In the case of closed-source models, it may be that access to the model is limited in some way (e.g., to registered users), but it should be possible for other researchers to have some path to reproducing or verifying the results.
        \end{enumerate}
    \end{itemize}

\item {\bf Open access to data and code}
    \item[] Question: Does the paper provide open access to the data and code, with sufficient instructions to faithfully reproduce the main experimental results, as described in supplemental material?
    \item[] Answer: \answerYes{} 
    \item[] Justification: \weprovidecode
    \item[] Guidelines:
    \begin{itemize}
        \item The answer NA means that paper does not include experiments requiring code.
        \item Please see the NeurIPS code and data submission guidelines (\url{https://nips.cc/public/guides/CodeSubmissionPolicy}) for more details.
        \item While we encourage the release of code and data, we understand that this might not be possible, so “No” is an acceptable answer. Papers cannot be rejected simply for not including code, unless this is central to the contribution (e.g., for a new open-source benchmark).
        \item The instructions should contain the exact command and environment needed to run to reproduce the results. See the NeurIPS code and data submission guidelines (\url{https://nips.cc/public/guides/CodeSubmissionPolicy}) for more details.
        \item The authors should provide instructions on data access and preparation, including how to access the raw data, preprocessed data, intermediate data, and generated data, etc.
        \item The authors should provide scripts to reproduce all experimental results for the new proposed method and baselines. If only a subset of experiments are reproducible, they should state which ones are omitted from the script and why.
        \item At submission time, to preserve anonymity, the authors should release anonymized versions (if applicable).
        \item Providing as much information as possible in supplemental material (appended to the paper) is recommended, but including URLs to data and code is permitted.
    \end{itemize}

\item {\bf Experimental setting/details}
    \item[] Question: Does the paper specify all the training and test details (e.g., data splits, hyperparameters, how they were chosen, type of optimizer, etc.) necessary to understand the results?
    \item[] Answer: \answerYes{} 
    \item[] Guidelines:
    \begin{itemize}
        \item The answer NA means that the paper does not include experiments.
        \item The experimental setting should be presented in the core of the paper to a level of detail that is necessary to appreciate the results and make sense of them.
        \item The full details can be provided either with the code, in appendix, or as supplemental material.
    \end{itemize}

\item {\bf Experiment statistical significance}
    \item[] Question: Does the paper report error bars suitably and correctly defined or other appropriate information about the statistical significance of the experiments?
    \item[] Answer: \answerYes{} 
    \item[] Justification: Yes, we report sample standard deviation error bars in the plots (and in the result files in the supplementary material). These correspond to the randomness used by the algorithm.
    \item[] Guidelines:
    \begin{itemize}
        \item The answer NA means that the paper does not include experiments.
        \item The authors should answer "Yes" if the results are accompanied by error bars, confidence intervals, or statistical significance tests, at least for the experiments that support the main claims of the paper.
        \item The factors of variability that the error bars are capturing should be clearly stated (for example, train/test split, initialization, random drawing of some parameter, or overall run with given experimental conditions).
        \item The method for calculating the error bars should be explained (closed form formula, call to a library function, bootstrap, etc.)
        \item The assumptions made should be given (e.g., Normally distributed errors).
        \item It should be clear whether the error bar is the standard deviation or the standard error of the mean.
        \item It is OK to report 1-sigma error bars, but one should state it. The authors should preferably report a 2-sigma error bar than state that they have a 96\% CI, if the hypothesis of Normality of errors is not verified.
        \item For asymmetric distributions, the authors should be careful not to show in tables or figures symmetric error bars that would yield results that are out of range (e.g. negative error rates).
        \item If error bars are reported in tables or plots, The authors should explain in the text how they were calculated and reference the corresponding figures or tables in the text.
    \end{itemize}

\item {\bf Experiments compute resources}
    \item[] Question: For each experiment, does the paper provide sufficient information on the computer resources (type of compute workers, memory, time of execution) needed to reproduce the experiments?
    \item[] Answer: \answerYes{} 
    \item[] Justification: The experiments can be run on commodity hardware (CPU only, single-threaded) and take several hours to finish. We mention this in \cref{sec:experiments}.
    \item[] Guidelines:
    \begin{itemize}
        \item The answer NA means that the paper does not include experiments.
        \item The paper should indicate the type of compute workers CPU or GPU, internal cluster, or cloud provider, including relevant memory and storage.
        \item The paper should provide the amount of compute required for each of the individual experimental runs as well as estimate the total compute. 
        \item The paper should disclose whether the full research project required more compute than the experiments reported in the paper (e.g., preliminary or failed experiments that didn't make it into the paper). 
    \end{itemize}
    
\item {\bf Code of ethics}
    \item[] Question: Does the research conducted in the paper conform, in every respect, with the NeurIPS Code of Ethics \url{https://neurips.cc/public/EthicsGuidelines}?
    \item[] Answer: \answerYes{} 
    \item[] Guidelines:
    \begin{itemize}
        \item The answer NA means that the authors have not reviewed the NeurIPS Code of Ethics.
        \item If the authors answer No, they should explain the special circumstances that require a deviation from the Code of Ethics.
        \item The authors should make sure to preserve anonymity (e.g., if there is a special consideration due to laws or regulations in their jurisdiction).
    \end{itemize}

\item {\bf Broader impacts}
    \item[] Question: Does the paper discuss both potential positive societal impacts and negative societal impacts of the work performed?
    \item[] Answer: \answerYes{} 
    \item[] Justification: Yes, in \cref{sec:conclusion}.
    \item[] Guidelines:
    \begin{itemize}
        \item The answer NA means that there is no societal impact of the work performed.
        \item If the authors answer NA or No, they should explain why their work has no societal impact or why the paper does not address societal impact.
        \item Examples of negative societal impacts include potential malicious or unintended uses (e.g., disinformation, generating fake profiles, surveillance), fairness considerations (e.g., deployment of technologies that could make decisions that unfairly impact specific groups), privacy considerations, and security considerations.
        \item The conference expects that many papers will be foundational research and not tied to particular applications, let alone deployments. However, if there is a direct path to any negative applications, the authors should point it out. For example, it is legitimate to point out that an improvement in the quality of generative models could be used to generate deepfakes for disinformation. On the other hand, it is not needed to point out that a generic algorithm for optimizing neural networks could enable people to train models that generate Deepfakes faster.
        \item The authors should consider possible harms that could arise when the technology is being used as intended and functioning correctly, harms that could arise when the technology is being used as intended but gives incorrect results, and harms following from (intentional or unintentional) misuse of the technology.
        \item If there are negative societal impacts, the authors could also discuss possible mitigation strategies (e.g., gated release of models, providing defenses in addition to attacks, mechanisms for monitoring misuse, mechanisms to monitor how a system learns from feedback over time, improving the efficiency and accessibility of ML).
    \end{itemize}
    
\item {\bf Safeguards}
    \item[] Question: Does the paper describe safeguards that have been put in place for responsible release of data or models that have a high risk for misuse (e.g., pretrained language models, image generators, or scraped datasets)?
    \item[] Answer: \answerNA{} 
    \item[] Guidelines:
    \begin{itemize}
        \item The answer NA means that the paper poses no such risks.
        \item Released models that have a high risk for misuse or dual-use should be released with necessary safeguards to allow for controlled use of the model, for example by requiring that users adhere to usage guidelines or restrictions to access the model or implementing safety filters. 
        \item Datasets that have been scraped from the Internet could pose safety risks. The authors should describe how they avoided releasing unsafe images.
        \item We recognize that providing effective safeguards is challenging, and many papers do not require this, but we encourage authors to take this into account and make a best faith effort.
    \end{itemize}

\item {\bf Licenses for existing assets}
    \item[] Question: Are the creators or original owners of assets (e.g., code, data, models), used in the paper, properly credited and are the license and terms of use explicitly mentioned and properly respected?
    \item[] Answer: \answerYes{} 
    \item[] Guidelines:
    \begin{itemize}
        \item The answer NA means that the paper does not use existing assets.
        \item The authors should cite the original paper that produced the code package or dataset.
        \item The authors should state which version of the asset is used and, if possible, include a URL.
        \item The name of the license (e.g., CC-BY 4.0) should be included for each asset.
        \item For scraped data from a particular source (e.g., website), the copyright and terms of service of that source should be provided.
        \item If assets are released, the license, copyright information, and terms of use in the package should be provided. For popular datasets, \url{paperswithcode.com/datasets} has curated licenses for some datasets. Their licensing guide can help determine the license of a dataset.
        \item For existing datasets that are re-packaged, both the original license and the license of the derived asset (if it has changed) should be provided.
        \item If this information is not available online, the authors are encouraged to reach out to the asset's creators.
    \end{itemize}

\item {\bf New assets}
    \item[] Question: Are new assets introduced in the paper well documented and is the documentation provided alongside the assets?
    \item[] Answer: \answerNA{} 
    \item[] Justification: We release no new assets (except code, which is discussed elsewhere).
    \item[] Guidelines:
    \begin{itemize}
        \item The answer NA means that the paper does not release new assets.
        \item Researchers should communicate the details of the dataset/code/model as part of their submissions via structured templates. This includes details about training, license, limitations, etc. 
        \item The paper should discuss whether and how consent was obtained from people whose asset is used.
        \item At submission time, remember to anonymize your assets (if applicable). You can either create an anonymized URL or include an anonymized zip file.
    \end{itemize}

\item {\bf Crowdsourcing and research with human subjects}
    \item[] Question: For crowdsourcing experiments and research with human subjects, does the paper include the full text of instructions given to participants and screenshots, if applicable, as well as details about compensation (if any)? 
    \item[] Answer: \answerNA{} 
    \item[] Guidelines:
    \begin{itemize}
        \item The answer NA means that the paper does not involve crowdsourcing nor research with human subjects.
        \item Including this information in the supplemental material is fine, but if the main contribution of the paper involves human subjects, then as much detail as possible should be included in the main paper. 
        \item According to the NeurIPS Code of Ethics, workers involved in data collection, curation, or other labor should be paid at least the minimum wage in the country of the data collector. 
    \end{itemize}

\item {\bf Institutional review board (IRB) approvals or equivalent for research with human subjects}
    \item[] Question: Does the paper describe potential risks incurred by study participants, whether such risks were disclosed to the subjects, and whether Institutional Review Board (IRB) approvals (or an equivalent approval/review based on the requirements of your country or institution) were obtained?
    \item[] Answer: \answerNA{} 
    \item[] Guidelines:
    \begin{itemize}
        \item The answer NA means that the paper does not involve crowdsourcing nor research with human subjects.
        \item Depending on the country in which research is conducted, IRB approval (or equivalent) may be required for any human subjects research. If you obtained IRB approval, you should clearly state this in the paper. 
        \item We recognize that the procedures for this may vary significantly between institutions and locations, and we expect authors to adhere to the NeurIPS Code of Ethics and the guidelines for their institution. 
        \item For initial submissions, do not include any information that would break anonymity (if applicable), such as the institution conducting the review.
    \end{itemize}

\item {\bf Declaration of LLM usage}
    \item[] Question: Does the paper describe the usage of LLMs if it is an important, original, or non-standard component of the core methods in this research? Note that if the LLM is used only for writing, editing, or formatting purposes and does not impact the core methodology, scientific rigorousness, or originality of the research, declaration is not required.
    \item[] Answer: \answerNA{} 
    \item[] Guidelines:
    \begin{itemize}
        \item The answer NA means that the core method development in this research does not involve LLMs as any important, original, or non-standard components.
        \item Please refer to our LLM policy (\url{https://neurips.cc/Conferences/2025/LLM}) for what should or should not be described.
    \end{itemize}

\end{enumerate}

\end{document}